\documentclass[11pt,letterpaper]{article}  
\usepackage[margin=1in]{geometry}
\usepackage{hyperref}
\usepackage{stmaryrd} 
\usepackage[normalem]{ulem}
\usepackage{xspace,graphicx,amssymb,amsmath,amsthm,thm-restate} 
\usepackage{algorithm}  
\usepackage[]{algpseudocode} 
\usepackage{caption} 
\usepackage[capitalise]{cleveref} 
\Crefname{figure}{Figure}{Figures}
\usepackage{url} 

\newtheorem{theorem}{Theorem}
\newtheorem{lemma}[theorem]{Lemma}
\theoremstyle{definition}  
\newtheorem{Definition}[theorem]{Definition}
\newtheorem{note}[theorem]{Remark} 
\newtheorem{remark}[theorem]{Remark}

\newcommand{\snMFE}{snMFE\xspace}
\newcommand{\symnMFE}{symmetry-naive MFE\xspace}
\newcommand{\SymnMFE}{Symmetry-naive MFE\xspace}
\newcommand{\PolySpi}{\ensuremath{\mathrm{Poly}(S,\pi)}\xspace}
\newcommand{\PolySpiPrime}{\ensuremath{\mathrm{Poly}(S,\pi')}\xspace}

\newcommand{\DGnosym}{\ensuremath{\overline{\Delta G}}}
\newcommand{\algorithmicbreak}{\textbf{break}}
\newcommand{\Break}{\State \algorithmicbreak}

\usepackage[textsize=tiny,
textcolor=blue,
linecolor=blue,
bordercolor=blue,
backgroundcolor=white,
]{todonotes}

\title{\Large An efficient algorithm to compute the minimum free energy\\ of interacting nucleic acid strands
}

\author{Ahmed Shalaby\thanks{Hamilton Institute and Department of Computer Science, Maynooth University, Ireland. Research supported by Science Foundation Ireland (SFI) under grant numbers 20/FFP-P/8843 and 18/ERCS/5746 and the European Union's European Research Council (ERC, Active-DNA, no 772766); European Innovation Council  (EIC, DISCO, No 101115422). Views and opinions expressed are however those of the author(s) only and do not necessarily reflect those of the European Union, European Research Council,  European Innovation Council or Science Foundation Ireland. Neither the European Union nor the granting authority can be held responsible for them.} 
\and Damien Woods\footnotemark[1]}

\date{}

\begin{document}
	\maketitle
	\pagenumbering{arabic}\pagestyle{plain}
	
\begin{abstract} \baselineskip=11pt 
	The information-encoding molecules RNA and DNA form a combinatorially large set of secondary structures through nucleic acid base pairing.  Thermodynamic prediction algorithms predict favoured, or minimum free energy (MFE), secondary structures, and can even assign an equilibrium probability to any particular structure via the partition function---a Boltzmann-weighted sum of the free energies of the exponentially large set of secondary structures.  Prediction is NP-hard in the presence pseudoknots---base pairings that violate a restricted planarity condition.  However, unpseudoknotted structures are amenable to dynamic programming-style problem decomposition:  for a single DNA/RNA strand there are polynomial time algorithms for MFE and partition function.  For multiple strands, the problem is significantly more complicated due to extra entropic penalties.  Dirks et al [SICOMP Review; 2007] showed that for multiple ($\mathcal{O}(1)$) strands, with $N$ bases, there is a polynomial time in $N$ partition function algorithm, however their technique did not generalise to  MFE which they left open. 
	
	We give the first polynomial time ($\mathcal{O}(N^4)$) algorithm   for unpseudoknotted multiple ($\mathcal{O}(1)$) strand MFE, answering the open problem from Dirks et al.  The challenge in computing MFE lies in considering the rotational symmetry of secondary structures, a global feature not immediately amenable to dynamic programming algorithms that rely on local subproblem decomposition.  Our proof has two main technical contributions:  First, a polynomial upper bound on the number of symmetric secondary structures that need to be considered when computing the rotational symmetry penalty. Second, that bound is leveraged by a backtracking algorithm to find the true MFE in an exponential space of contenders. 
	
	Our MFE algorithm has the same asymptotic run time as Dirks et al's partition function algorithm,  suggesting a reasonably efficient handling of the global problem of rotational symmetry, although ours has higher space complexity.  Finally, our algorithm also seems reasonably tight in terms of number of strands since Codon, Hajiaghayi and Thachuk [DNA27, 2021] have shown  that unpseudoknotted MFE is NP-hard for $\mathcal{O}(N)$ strands. 
\end{abstract}

\newcommand{\base}[1]{\ensuremath{\mathrm{#1}}\xspace}
\newcommand{\baseA}{\base{A}}
\newcommand{\baseT}{\base{T}}
\newcommand{\baseG}{\base{G}}
\newcommand{\baseC}{\base{C}}
\newcommand{\baseU}{\base{U}}

\section{Introduction}

The  {\em primary structure} of a DNA strand is simply a word   over the alphabet  $\{\baseA, \baseC, \baseG, \baseT\}$, or  $\{\baseA, \baseC, \baseG, \baseU\}$ for RNA.  
Bases may bond in pairs, \baseA binds to \baseT and \baseC binds to \baseG, and a set of such pairings for a strand is called a {\em secondary structure} as shown in {\cref{fig:sec struct}(a)}; typically each strand has exponentially many possible secondary structures.\footnote{A secondary structure, along with a set of experimental conditions,  induces one or more 3D structures called tertiary structures---a complication we will not be concerned with in this paper since, unlike proteins, it is fortunate that DNA/RNA interactions are sufficiently chemically simple that that somewhat elementary secondary structure model is sufficient for useful structural prediction.} 
Mainly, what practitioners care about are probabilities of a given secondary structure or class of secondary structures. 
For that, each secondary structure~$s$ has an associated, typically negative, real valued {\em free energy} $\Delta G(s)$, where more negative is deemed more favourable.  
Thus the most favourable is the secondary structure, or structures, with minimum free energy (MFE). 
More generally, the Boltzman distribution is  a probability distribution on secondary structures $s$ at chemical equilibrium:  
the probability of~$s$ is  $p(s) = \frac{1}{Z} \mathrm{e}^{- \Delta G(s)/k_\mathrm{B}T}  $ where $Z$ is a normalisation factor called the partition function: 
\begin{equation}\label{eq:pf}
	Z  = \sum_{s\in\Omega} \mathrm{e}^{- \Delta G(s)/k_\mathrm{B}T} 
\end{equation}
that is, an exponentially weighted sum of the free energies over the set~$\Omega$ of all secondary structures, 
where~$k_\mathrm{B}$ is Boltzmann's constant and $T$ is temperature in Kelvin. 

\begin{table}[t]
	\centering
	\begin{tabular}{ p{5cm}||p{5cm}|p{5cm}  }
		
		Input type& MFE & Partition function\\ 
		
		\hline\hline
		
		Single strand   &   $\mathcal{O}(N^4)$ \ \   \cite{zukeroptimal,zukerrna, nussinov1978algorithms, nussinov1980fast,waterman1986rapid}  &   $\mathcal{O}(N^4)$ \ \  \cite{mccaskill1990equilibrium}     \\  \hline
		Multiple strands, bounded,\newline i.e.~$c= \mathcal{O}(1)$ strands &  $\mathcal{O}(N^4 (c-1)!)$ \ \  {\bf [Theorem~\ref{thm:main}]}   &   $\mathcal{O}(N^4 (c-1)!)$  \ \ \  \cite{dirks2007thermodynamic}\footref{ft:N3}    \\ \hline
		Multiple strands, unbounded,\newline i.e.~$\mathcal{O}(N)$ strands &  APX-hard \ \   \cite{condon2021predicting}  &  Open problem   \\ \hline
	\end{tabular}\vspace{-1.5ex}
	\caption{\label{table}\small
		Some algorithmic results for MFE and partition function for unpseudoknoted\footref{ft:pseudoknot} nucleic acid systems.
		$N$ is the total number of bases of all strand(s) in the system (i.e.~sum of strand lengths).
		Results are shown for input being a single strand, 
		multiple strands bounded by a constant or unbounded/growing with~$N$. 
		Note that in the literature the polynomial is sometimes written to the power 3 (e.g.~$\mathcal{O}(N^3 ...)$), but this is for a restricted ``relaxation'' of the model.\footref{ft:N3}  
	}
\end{table}

Decades ago, the deep relationship between secondary structures and dynamic programming algorithms was established~\cite{zukeroptimal,zukerrna, nussinov1978algorithms, nussinov1980fast,waterman1986rapid, mccaskill1990equilibrium}.  
If a secondary structure can be drawn as a polymer graph  without edge crossings it is called unpseudoknotted  (\cref{fig:sec struct}(c)).
The earliest polynomial time algorithms were for single-stranded  unpseudoknotted secondary structures, with the absence of crossings allowing for planar decompositions of secondary structures that are suited  to dynamic programming techniques.\footnote{Exclusion of pseudoknots is usually founded on both modelling and algorithmic considerations. Energy models for pseudoknots are difficult to formulate due to the increased significance of geometric issues and tertiary interactions \cite{dirks2007thermodynamic}. If pseudoknots are permitted, it is known that the MFE prediction is NP-hard even for a single strand \cite{akutsu2000dynamic,lyngso2000pseudoknots,lyngso2000rna}.The first NP-hardness results \cite{akutsu2000dynamic,lyngso2000pseudoknots} used a simple energy model called the stacking model where only consecutive base pairs forming a stack contribute to the free energy of a secondary structure. These hardness results with relatively simple energy models, make it seems unlikely that the MFE prediction problem will be easier in the case of more complicated energy models~\cite{condon2021predicting}. But, dynamic programming algorithms are still possible for restricted classes of pseudoknots, for both MFE prediction \cite{rivas1999dynamic, uemura1999tree, chen2009n, jabbari2018knotty, reeder2004design} and partition function \cite{dirks2003partition, dirks2004algorithm}. \label{ft:pseudoknot}}  
For a single RNA/DNA strand, both MFE and partition function are computable in $\mathcal{O}(N^4)$ time (\cref{table}),  using the standard  energy model\footnote{This model is variably called the nearest neighbour model, the Turner model, and loop energy model. 
Versions of the model have been implemented in software suites such as NUPACK~\cite{dirks2007thermodynamic,dirks2004algorithm,fornace2020unified}, ViennaRNA~\cite{viennaRNA} and mfold~\cite{mfold}, for both RNA and DNA~\cite{santalucia1998unified,santa}.} that will be formally defined in \cref{sec:mfe}. 

Work in DNA computing~\cite{algoSST,squareRoot,cargoSorting,SIMDDNA,celltimer2017fernshulman,Chatterjee2017,seelig2006enzyme,zhang2011dynamic}, and nucleic acid nanotechnology more generally~\cite{geary2014single,woo2011programmable}, involves building molecular systems and structures with, to date, hundreds, and soon, thousands, of interacting strands, so there is a need for better algorithms for these multi-stranded `inverse design problems'~\cite{churkin2018design,nuad}.  
And, of course, biologists need to understand molecular structure in order to understand and predict molecular interactions. 
However, when there are multiple interacting strands, the situation becomes significantly more complicated than the single-stranded case for two reasons:  
First, for a secondary structure to be unpseudoknotted, it implies there should be {\em at least one} permutation of the strands without crossings on the polymer graph~\cite{dirks2007thermodynamic} (\cref{fig:sec struct}). 
Second, if strand types are repeated then 
so-called {\em rotational symmetries} (\cref{fig:sym}) arise that need to be accounted for in the model to match the underlying statistical mechanics\footnote{This fact from statistical mechanics is discussed in some papers~\cite{dirks2007thermodynamic,hofacker2012symmetric}, although we've not found its full derivation in the modern nucleic-acid algorithmic literature. We leave a first-principles derivation for future work.\label{ft:statmech}}, otherwise structures may be over- or undercounted, leading to incorrect probabilities in the Boltzmann distribution, in other words: incorrect predicted free energy of a secondary structure.


For multiple strands, albeit a constant number $c = \mathcal{O}(1)$, 
Dirks,  Bois, Schaeffer, Winfree and  
Pierce~\cite{dirks2007thermodynamic} gave a polynomial time partition function algorithm running in time $\mathcal{O}(N^4 (c-1)!)$.\footnote{We note that Dirks et al~\cite{dirks2007thermodynamic}, and others in field~\cite{zukeroptimal,zukerrna,condon2021predicting,waterman1986rapid,mccaskill1990equilibrium,boehmer2024rna}, often state the run time with 3  instead of 4 in the exponent (i.e. $\mathcal{O}(N^3)$ for single stranded and $\mathcal{O}(N^3 (c-1)!)$ for multistranded). This reduction comes from changing the standard energy model by putting some restrictions on the size of interior loops (\cref{fig:sec struct}), or by enforcing certain mild conditions on the energy parameters for the interior loops  \cite{lyngso1999fast,hofacker2012symmetric}, we do not assume these additional assumptions in this work.\label{ft:N3} Note that our time/space complexity could benefit from such model changes for example, via reduction of the upper bound in \cref{lem:polyub}}. 
The first problem goes away by simply assuming $c$ is a constant;
but if the number of strands is non-constant, in particular $c = \mathcal{O}(N)$, 
Codon, Hajiaghayi and Thachuk~\cite{condon2021predicting} showed MFE is NP-hard, and even APX-hard.\footnote{This hardness result holds whether or not rotational symmetries are accounted for in the energy model.} 
For the second, rotational symmetry, problem, in order to compute partition function, Dirks et al~\cite{dirks2007thermodynamic} found an algebraic link between the overcounting and the rotational symmetry correction problems, which allowed both to be solved simultaneously, aided by the exponential nature of the partition function.  
Surprisingly, that trick does not work for MFE:
Since MFE prediction is minimization-based, 
there is no secondary structure overcounting problem in MFE prediction---repeated secondary structures will not change the outcome of minimization, unlike the partition function which is summation-based.
Hence, the absence of the overcounting problem makes MFE prediction harder to solve, and was left open by Dirks et al~\cite{dirks2007thermodynamic}.
For the special case of two strands, Hofacker, Reidys, and Stadler~\cite{hofacker2012symmetric} gave an $\mathcal{O}(N^6)$  algorithm. In this paper, we propose an efficient solution to the $\mathcal{O}(1)$ strand MFE problem, the first that runs in polynomial time.

\subsection{Statement of main result}
Our main result is the following theorem, whose proof is in \cref{sec:analysis}: 

\begin{restatable}{theorem}{main} 
	\label{thm:main}
	There is an $\mathcal{O}(N^4(c-1)!)$ time and $\mathcal{O}(N^4)$ space algorithm for the 
	Minimum Free Energy unpseudoknotted secondary structure prediction problem, including rotational symmetry, 
	for a set of $c = \mathcal{O}(1)$ DNA or RNA strands of  total length  $N$ bases. 
\end{restatable}

In \cref{sec:analysis} we give a time-space trade-off for our result, by showing a variation of the algorithm runs in $\mathcal{O}(N^4 \log N (c-1)!)$ time but $\mathcal{O}(N^3)$ space.

We use the standard~\cite{dirks2007thermodynamic} definition of free energy (\cref{eq:DGss}) of multistranded unpseudoknotted secondary structures, which includes rotational symmetry, see \cref{sec:mfe} for formal definitions.  
We first give an extensive overview of the proof and paper structure, followed by future work.

\begin{figure}[t]
	\centering\includegraphics[width=0.9\textwidth]{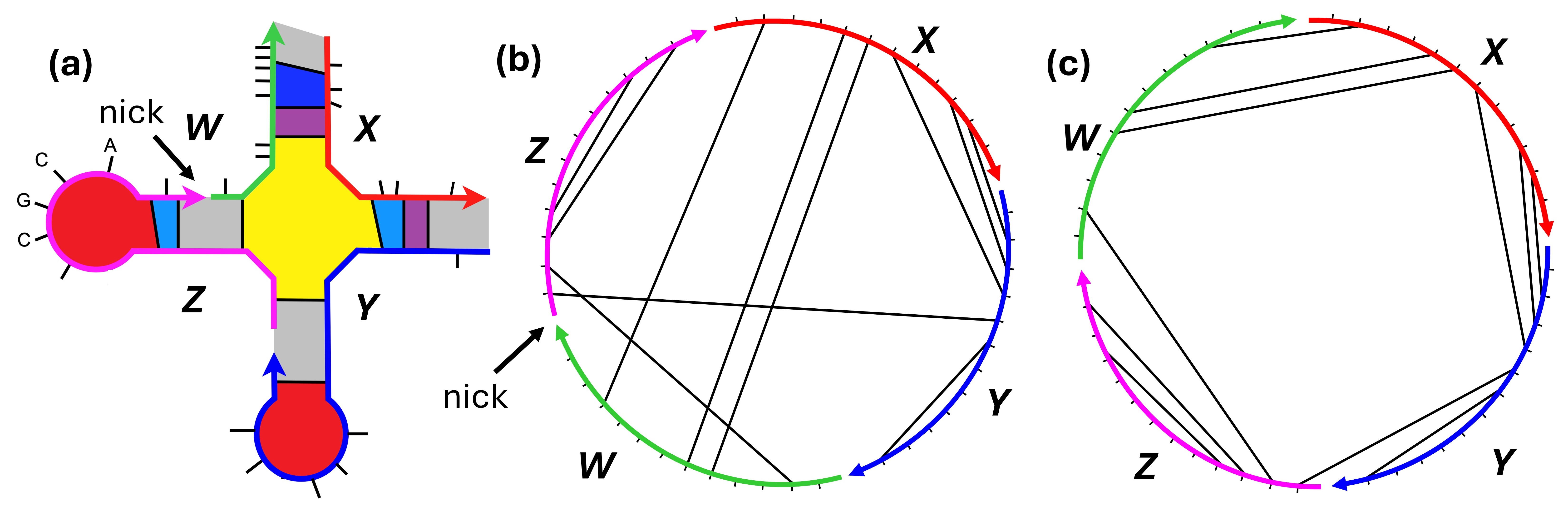}
	\caption{
		A  DNA (or RNA) secondary structure $S$ with $c=4$ strands and two of its $(c-1)!=6$ polymer graphs. 
		(a) One of the many possible secondary structures for four DNA strands $W,X,Y,Z$. 
		Short black lines represent DNA bases (a few are shown $\ldots \baseC, \baseG, \baseC, \baseA \ldots$), and long lines represent base pairs (drawing not to scale). 
		Loops are colour-coded as follows:  stack=purple, multiloop=yellow, hairpin=red, bulge=light blue, internal=dark blue, external=grey.    
		Black arrow: the small gap between two strands is called a {\em nick}. 
		(b)~Polymer graph for the strand ordering $\pi' = WZXY$, denoted \PolySpiPrime, showing base-pair crossings.
		(c)~By  reordering to $\pi = WXYZ$ we get another polymer graph \PolySpi for $S$, without crossings, hence  $S$ is unpseudoknotted. 
	}
	\label{fig:sec struct}
	
\end{figure}

\subsection{Proof overview and paper structure}\label{sec:intuition}

\subsubsection{The main challenge: handling rotational symmetry}

Typically, each DNA base pair that forms represent a decrease (improvement in favourability) in free energy---although not always. 
In a multi-stranded system, when several strands bind together the entropy of the overall system is decreased since there are now less states due to their being less free molecules.  
Thus the energy model for multistranded systems includes an entropic {\em association penalty} (typically positive) for every extra strand, beyond the first, bound into a multistranded molecular  complex~\cite{dirks2007thermodynamic}. 
However, statistical mechanics tells us to be careful about symmetry: with multiple identical strands in a complex it is possible that the complex is rotationally symmetric, 
intuitively  there are several complexes, identical up to rotation of their polymer graphs (\cref{fig:sym}).\footnote{Formally, we mean the permutation representing the complex is rotationally symmetric.} 
These so-called indistinguishable complexes, in turn imply that a (positive) penalty should be applied to account for the difference in entropy between a similar, but distinguishable, complex without  rotational symmetry~\cite{bormashenko2019entropy,atkins2023atkins,silbey2022physical,fornace2020unified}.\footref{ft:statmech} 
\cref{sec:mfe} gives definitions needed  to formalise these concepts, including: DNA,  
unpseudoknotted secondary structure, polymer graph,  
free energy including rotational symmetry (\cref{eq:DGss}) and MFE (\cref{eq:MFE}). In particular, \cref{sec:sym} gives a  group-theoretic definition of rotational symmetry, to help formalise some of the prior work.

\subsubsection{General approach to find the true MFE}
One obvious idea might be to find a dynamic programming algorithm that directly handles rotational symmetry. 
However, this approache suffers from rotational symmetry being a {\em global property} of an entire system state (secondary structure), whereas dynamic programming relies on piecing together subproblems that are individually unaware of the global context---or more precisely, may be used in multiple global contexts whether symmetric or not.  

Instead, our strategy is to first compute what we call the {\em \symnMFE} (\snMFE) that (incorrectly) assumes all strands are distinct and thus does not compute correct free energies for rotational symmetries. 
We use  Dirks et al's \snMFE algorithm~\cite{dirks2007thermodynamic},  that assumes all strands are distinct, but  augmented to return extra dynamic programming matrices (\cref{algo:1} in \cref{sec:AlgoMFE}). 
We use these extra matrices to compute the required symmetry correction to that \snMFE value using a backtracking algorithm, as  follows.

\subsubsection{Polynomial upper bound: intuition for \cref{sec:ub}} 
Our goal is to show that, after running our augmentation of the known algorithm for \snMFE, we have implicit access to a collection of secondary structures that are `not too far' from the true MFE---where by `not too far' we mean we have a polynomial bound on the number of structures to be  considered by another fast (``backtracking'') algorithm that finds the true MFE structure.

First, 
to see how we find this polynomial bound,  imagine  the augmented
\snMFE algorithm finds that the secondary structure with \snMFE is  rotationally {\em  asymmetric}, hence we are done, we know that the \snMFE value is in fact the true MFE. 
Otherwise, we have a {\em rotationally symmetric} secondary structure: ideally we would like to compute it's rotational symmetry degree $R$ (takes linear time in the size of the secondary structure) and then return  $\mathrm{\snMFE} + k_\mathrm{B} T \log R$ as the true MFE, but this approach is doomed to fail since there could also be structures with lower true MFE, i.e.~in the real interval $[\mathrm{\snMFE}, \ \mathrm{\snMFE} + k_\mathrm{B} T \log R) \subset \mathbb{R}$.

Leveraging the two properties of being (a) unpseudoknotted and (b) rotationally symmetric, in \cref{sec:ub} we define a class of cuts of a structure's polymer graph (\cref{fig:sec struct}) that we call {\em pizza cuts}, or, more formally, \emph{admissible symmetric backbone cuts} (\cref{def:admissible}). 
These cuts are radially symmetric, hence the name pizza cut---how one slices a pizza from disk-edge to centre.
In \cref{lem:ub,lem:ub2}, we show that there are at most a polynomial number of pizza cuts that symmetric structures may have.

Then, when we do a backtracking-based search (below), 
through the dynamic programming matrices from the structure(s) with \snMFE, to larger free energies:  
if we find two different symmetric pizzas, but with the same pizza cuts,  we make a new pizza, by swapping a slice from one with a slice from the other.  
We prove that the new pizza is (a) guaranteed to be asymmetric and (b)~has free energy sandwiched between the {\snMFE} values of  two symmetric structures  (\cref{lem:sand,lem:ub2}). 
Moreover, this new structure's free energy is the true MFE. 
Otherwise we either find a symmetric structure (we output its naive free energy as the true MFE), or we reach a contradiction (i.e.~which can not happen) by reaching the polynomial bound having exhausted the set of all \emph{admissible symmetric backbone cuts}. In all cases the true MFE and its structure are output.

\subsubsection{Backtracking to find the true MFE: intuition for \cref{sec:BT}} 

It remains to show how we will do the backtracking search mentioned above. 
In \cref{sec:BT} we analyse the  backtracking algorithm, which is given in \cref{algo:2} in \cref{app:backalgo}, and  
is a polynomial time  algorithm over the exponentially large set of structures `close' to the true MFE value. 
It scans all secondary structures within an energy level starting with the \symnMFE (\snMFE) energy level, it goes on to sequentially scan higher levels in low-to-high order. 
The scanning process at any energy level $\mathcal{E}$  guarantees that each secondary structure that belongs to $\mathcal{E}$ should be scanned exactly once. 

The backtracking algorithm will run until one of the following conditions occurs:
(1) It scans an asymmetric secondary structure $S$, or 
(2) it exceeds the polynomial upper bound $\mathcal{U}$ of the number of symmetric secondary structures (i.e. the number of distinct  pizza cuts) to be scanned, or 
(3) the backtracking will start scanning a new energy level $\mathcal{E}' > \mathcal{B}$, where $\mathcal{B}$ is the current best candidate for MFE (the starting value for $\mathcal{B}$ is  $\mathcal{B} = \mathrm{\snMFE} + k_\mathrm{B} \log v(\pi)$ where $v(\pi)$ is the highest degree of rotational symmetry, \cref{def:sym}). 
Then, based on the condition that will occur, the algorithm directly returns the true MFE, and a secondary structure which has the true MFE will also be constructed. 
The short proof of \cref{thm:main} in \cref{sec:analysis} ties these results together to give the final analysis of our main result.

\subsection{Future work}
Our algorithm runs in polynomial time $\mathcal{O}(N^4 (c-1)!)$ for the case of $c=\mathcal{O}(1)$ strands, the $(c-1)!$ term coming from the fact that our algorithm, as well as Dirks et al~\cite{dirks2007thermodynamic}, is assumed to be called from an outer  loop that explicitly tries all $(c-1)!$ cyclic strand permutations. 
Can we increase the number of strands and still have a polynomial time algorithm? We know ``not by much'', since the problem is NP-complete when $c=\mathcal{O}(N)$~\cite{condon2021predicting}. 
Interestingly, Boehmer, Berkemer, Will, and Ponty~\cite{boehmer2024rna} recently reduced the $c$-strand parameterized running time for computing \symnMFE (i.e.~ignoring rotational symmetry) and partition function from factorial, $\mathcal{O}(N^4 (c-1)!)$, to exponential, $\mathcal{O}(N^4  3^c)$.\footref{ft:N3}  
It is feasible that our MFE algorithm could be augmented via this result.

Our MFE algorithm exploits a polynomial upper bound, $\mathcal{U}$, on the number of so-called symmetric secondary structures, or distinct  pizza cuts.  That bound is linear in ``most'' cases (\cref{lem:even}), but quadratic in one special subcase (\cref{lem:ub2}) of 2-fold rotational symmetry with a central internal loop. Reducing  that special case to linear would subtract one from  our algorithm's running time exponent.

\cref{table} shows the open problem for partition function on multiple strands. Intuitively, it seems that  partition function should be at least as hard as MFE, however that intuition is tempered by the fact that Dirks et al's approach for partition function did not carry over to MFE, hence this is an open problem. Indeed, our paper brings extra techniques to handle multi-stranded MFE.   

More generally, the computational complexity of partition function for DNA/RNA strands is less well understood than MFE. 
For example, are there settings where partition function, or problems counting numbers of structures, are {\#}P-complete?  

\section{Definition of multi-stranded DNA systems and basic lemmas}\label{sec:mfe}

Intuitively, a single DNA strand $s$ is a sequence of nucleotide bases connected by covalent bonds which together make up the backbone of $s$, with the left end of the sequence corresponding to the $5'$ end of $s$ and the right end corresponding to the $3'$ end. When drawing $s$ we label  the $3'$ end with an arrow which also shows the strand directionality, see \cref{fig:sec struct}. Hydrogen bonds can form between Watson-Crick base pairs, namely C–G and A–T.

Formally, A DNA strand $s$ is a word over the alphabet of DNA {\em bases} $\{\mathrm{A},\mathrm{T},\mathrm{G},\mathrm{C}\}$, indexed from 1 to $|s|$, where $|s|$ denotes the length of $s$.
A base pair is a tuple $(i, j)$ such that $i<j$. 
For any $c$ strands, we will assign to each of them a unique distinct identifier in $\{1, . . . ,c\}$~\cite{dirks2007thermodynamic}. Each base is specified by a strand identifier and a position on that strand, $i_s$ denotes the base of index $i$ of strand $s$.

\subsection{Connected unpseudoknotted secondary structures and  polymer graphs}

\begin{Definition}[Secondary structure $S$]
	For any set of $c$ DNA strands, a secondary structure $S$ is a set of base pairs such that each base appears in at most one pair, i.e.~if $(i_n, j_m)\in S$ and $(k_q, l_r)\in S$ then $i_n,j_m,k_q,l_r$ are all distinct.
\end{Definition}
The {\em graph representation of a  secondary structure} $S$
is the graph $G=(V,E)$, where $V$ is the set of bases of each strand $s \in \{1, . . . ,c\}$, and $E = E_v \cup E_b$, where $E_v$ is the set of \emph{covalent backbone bonds} connecting base $i_n$ with base $(i+1)_n$ for all bases $i = 1,2, ..., |n|-1$ on all strands $n \in \{1, . . . ,c\}$, and $E_b = S$ is the set of base pairs in $S$. $E_v$ and $E_b$ are disjoint.

The set of circular permutations, $\Pi$, of $c$ strands has $(c-1)!$ distinct circular permutations~\cite{brualdi1977introductory} (e.g., for the three strands $\{A, B, C\}$, $\Pi = \{A B C, ACB\}$), 
e.g., the orderings $ABC$, $BCA$, and $CAB$ are the same on a circle. 
Next, we define a  polymer graph for each $\pi$, see also  \cref{fig:sec struct}.

\begin{Definition}[Polymer graph]
	For any secondary structure $S$, and any ordering $\pi$ of its $c$ strands, the polymer graph representation of $S$,  denoted  \PolySpi, is a graph representation of $S$, embedded in the unit disk from $\mathbb{R}^2$, where the $c$ strands are placed in succession from their $5'$ to $3'$ ends around the circumference of the circle, and the bases, $V$, are spaced evenly around the circle circumference, each element of $E_v$ is represented by an arc on the circumference between covalently-bonded bases, and each element of $E_b$ is represented by a chord between two different bases. 
\end{Definition}

\begin{Definition}[Unpseudoknotted secondary structure]
	A secondary structure $S$ is unpseudoknotted if there exists at least one circular permutation $\pi \in \Pi$ such that $\PolySpi$ is planar, otherwise $S$ is pseudoknotted. 
	An example is shown in \cref{fig:sec struct}.
\end{Definition}

\begin{remark}
	In the rest of the paper we use $N$ to denote the total number of bases of a secondary structure $S$.
	A secondary structure $S$ is connected if the graph representation of $S$ is a connected graph. In this work, we are only interested in connected unpseudoknotted secondary structures.
\end{remark}

\begin{figure}[t]
	\centering\includegraphics[width=0.7\textwidth]{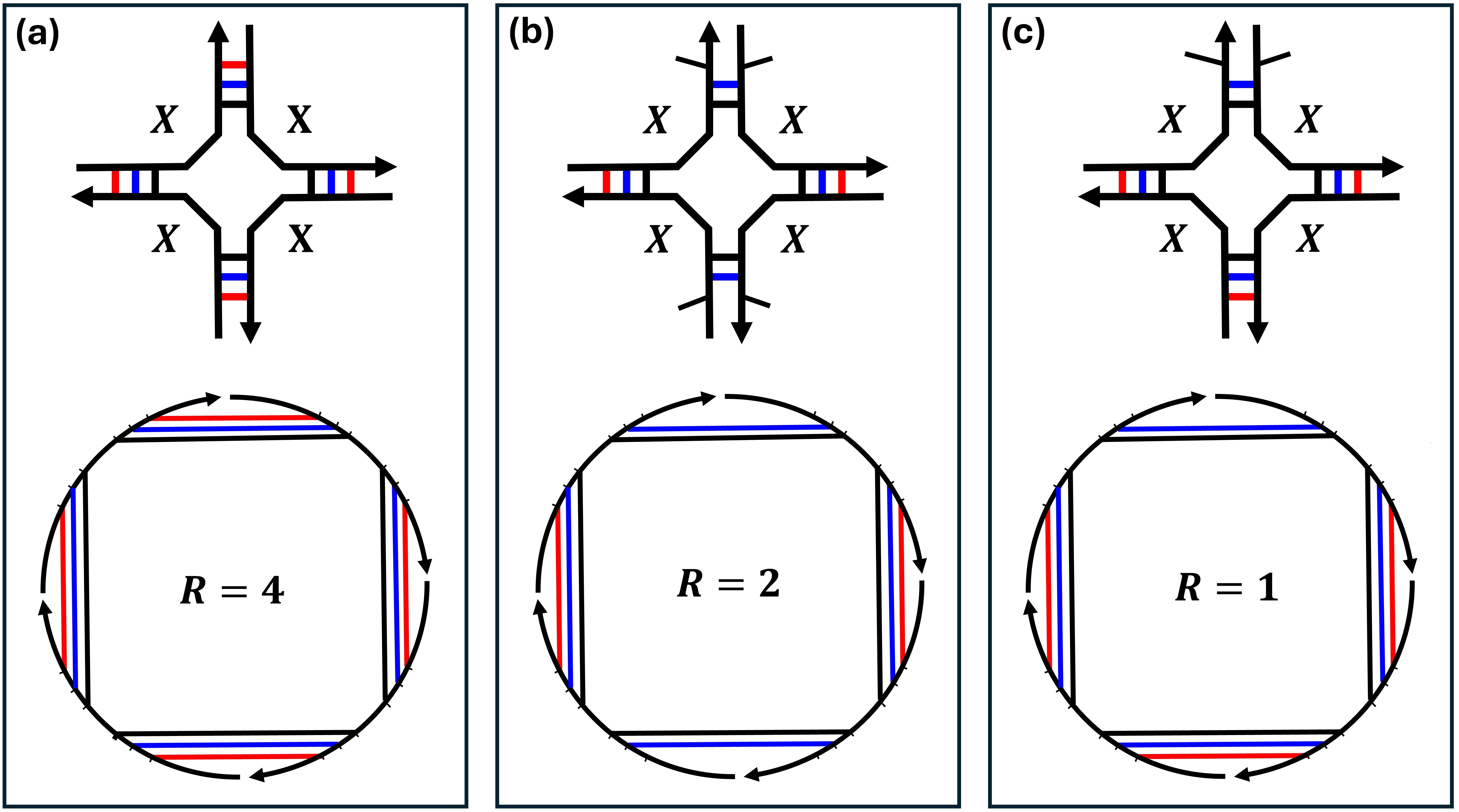}
	
	\caption{Three secondary structures  with their associated polymer graphs. In each case, there is a single complex with four identical (indistinguishable) strands of  strands of type $X$, but with different symmetry degree $R$. 
		(a)~Symmetry degree $R$ = 4 (rotation by $90^{\circ}$  gives the same secondary structure). 
		(b)~Symmetry degree  $R$ = 2 (rotation by $180^{\circ}$  gives the same secondary structure). 
		(c)~Symmetry degree  $R$ = 1 (asymmetric secondary structure).
	}\label{fig:sym}
\end{figure}

\subsection{Free energy of a secondary structure}
Any connected unpseudoknotted secondary structure $S$ can be decomposed into different loop types~\cite{tinoco,santa,mathews1999expanded}: namely hairpin loops, interior loops, exterior loops, stacks, bulges, and multiloops as shown in \cref{fig:sec struct}. 
As usual, let  $k_B$ be Boltzmann's constant and $T$ is the temperature in Kelvin (also a constant).\footnote{All results hold if we assume these are typical values from physics, or just 1 in appropriate units.} 
The free energy of $S$ is defined as the sum of three terms:\footnote{Throughout this paper $\log n = \log_\mathrm{e} n$.}
\begin{equation}\label{eq:DGss}
	\Delta G(S) =  \sum_{l\in S} \Delta G(l) + (c-1)\Delta G^{\textrm{assoc}} + k_\mathrm{B} T \log R.
\end{equation}  
\begin{itemize}
	\item the first is itself the sum of the (well-defined, empirically-obtained) free energies $\Delta G(l)$ of $S$'s constituent loops~\cite{dirks2007thermodynamic}, where each loop energy is defined with respect to the free energy of the unpaired reference state.
	\item $\Delta G^{\textrm{assoc}} $ is the entropic association~\cite{dirks2007thermodynamic} penalty applied for each of the $c-1$ strands added to the first strand to form a complex of $c$ strands.
	\item $R$ is the rotational symmetry of the secondary structure $S$,  illustrated in \cref{fig:sym}, and to be formally defined  in \cref{sec:sym}.  In particular, since favourable free energies are usually negative, the term  $k_\mathrm{B} T \log R \geq 0$ corresponds to the reduction in the entropic contribution of $S$, as any secondary structure with an $R$-fold rotational symmetry has a corresponding $R$-fold reduction in its distinguishable conformational space~\cite{dirks2007thermodynamic} as shown in \cref{fig:sym}.\footnote{This is perhaps counter-intuitive. In a follow-up expanded version of this paper we will give a full statistical mechanics explanation, which treats the symmetry penalty to offset the fact that non-symmetrical structures are undercounted.} Dynamic programming algorithms to date for mulit-stranded MFE ignored this term for reasons we outline in \cref{sec:sym}.
\end{itemize}

For $c$ strands, we let $\Omega$ be the set  (usually called the {\em ensemble}) of all  connected unpseudoknotted secondary structures. 
For any circular permutation $\pi \in \Pi$ of the $c$ strands, let $\Omega(\pi) \subseteq \Omega$ be the subset of $\Omega$ such that each connected unpseudoknotted secondary structure $S\in \Omega(\pi)$ is representable as a crossing-free polymer graph with circular permutation $\pi$. 

\begin{remark}[$S$, or $\PolySpi$]\label{rm:PolySpi}
	Dirks et al.~\cite{dirks2007thermodynamic} showed, in their representation theorem (Theorem 2.1), that the sets $\Omega(\pi)$, for all $\pi \in \Pi$, form a partitioning of $\Omega$, which means that every connected unpseudoknotted secondary structure belongs to exactly one $\Omega(\pi)$ for some $\pi \in \Pi$. 
	Hence to avoid the cumbersome phrase 
	{\em $c$-strand connected unpseudoknotted secondary structure $S$ with strand ordering $\pi$ and polymer graph $\PolySpi$}
	we simply write $S$, or $\PolySpi$.
\end{remark}

Predicting the minimum free energy means finding a minimum over the ensemble $\Omega$. 
The known strategy is to deal with each partition $\Omega(\pi)$ separately,  
then finding their minimum:  
\begin{equation}\label{eq:MFE}
	\textrm{MFE} = \min_{S \in \Omega} \Delta G(S)  =\min\limits_{\pi \in \Pi} \left\{ \min_{S \in \Omega(\pi)} \Delta G(S) \right\}  
\end{equation}

\subsection{Definition of multi-stranded rotational symmetry} \label{sec:sym}

Here, we formalise rotational symmetry. 
In the previous section we assigned each one of the $c$ strands a unique identifier, dealing with them as distinct strands even if two or more have the same sequence. 
But in most experimental settings, strands with the same sequences are \emph{indistinguishable}   in the sense that they behave identically with respect to relevant measurable quantities~\cite{dirks2007thermodynamic}. 
Mathematically, we say that {\em two strands are  indistinguishable} if they have the same sequence. 
Also, two \emph{secondary structures are indistinguishable} if there exists a permutation of the implied unique strand ordering (\cref{rm:PolySpi}), 
that maps indistinguishable strands onto each other while preserving all base pairs, otherwise, the two structures are distinct~\cite{dirks2007thermodynamic}.

For any $c$ strands, not necessarily distinct, they consist of $k \leq c$ \emph{strand types}, usually denoted by uppercase English letters $X,Y,\ldots$\footnote{In contrast with some of the literature. we exclude using $\{\mathrm{A},\mathrm{T},\mathrm{G},\mathrm{C}\}$ for strand types; to avoid any confusion between strand and base types.}
A {\em multi-stranded DNA system} 
$M = \{(t_1,n_1),(t_2,n_2), ..,(t_k,n_k)\}$,  
is a \emph{multiset} of $k$ strand types $t_1,  ..., t_k$ with repetition numbers 
$n_1,  ..., n_k \in \mathbb{N}$ such that $n_1+ ...+n_k = c$.\footnote{It is known~\cite{sawada2003fast} how to efficiently reduce the circular permutation space by getting rid of  circular permutations that are redundant due to indistinguishable strands, which is important to consider when computing the  partition function but needed for MFE.} 
For such a multiset $M$
we can think of each circular permutation $\pi$  as a string over strand types such that each strand type $t_i$ appears exactly $n_i$ times (e.g., $M = \{(X,6),(Z,3)\}$ one valid $\pi$ is $\pi = XZXXZXXZX$).

\begin{Definition}[Symmetry degree of a permutation] \label{def:sym}
	For any circular permutation $\pi$, we say $n \in \mathbb{N}$ is a symmetry degree of $\pi$ if $\pi = y^n$ for some $y$, a prefix of $\pi$. 
\end{Definition}

\noindent For example, $\{1,2,4\}$ are the symmetry degrees of $\pi = XZXZXZXZ$ since $\pi = (XZXZXZXZ)^1 = (XZXZ)^2 = (XZ)^4$.

For any circular permutation $\pi$, 
its maximum symmetry degree is denoted $v(\pi)$, and the corresponding repeating prefix $x$, such that $x^{v(\pi)} = \pi$, is the \emph{fundamental component} of $\pi$. 
It can be seen that $x$ is the smallest prefix that repeats over $\pi$. 
Indeed, $v(\pi)$ is the number of cyclic permutations that map each strand to a strand of the same type. 
Any repeating prefix of $\pi$ must be a multiple of its fundamental component, as proven in \cref{lem:factors} in \cref{sec:lemmasApp}.

\begin{remark}[Notation: $X_m^n$]
	For any circular permutation $\pi$, 
	its \emph{augmented version}  gives the full ordering information for each fundamental component. 
	For example, 
	if $\pi = XYXZ \, XYXZ$, then its fundamental component is $XYXZ$ and its augmented version is 
	$ X_1^1 Y_1^1 X_2^1 Z_1^1  \; X_1^2 Y_1^2 X_2^2 Z_1^2$, such that  $X_m^n$, means the $m$th strand of type $X$ in the $n$th fundamental component of $\pi$.   
\end{remark}

We can visualize any ordering $\pi$ by representing it as a \emph{regular $v(\pi)$-gon} with each of its $v(\pi)$ vertices   representing a fundamental component. 
Let $\rho = (1 \ 2 \ 3 \ ... \ v(\pi))$\footnote{Here we use algebraic cycle notation.  
	The order of $\rho$, denoted by $o(\rho)$, is the length of $\rho$ which is $v(\pi)$.} and consider the cyclic group\footnote{$G^\pi$ is isomorphic to $C_{v(\pi)}$, cyclic group of order $v(\pi)$.} $G^\pi$ generated by $\rho$. Intuitively, $G^\pi$ is the group of the all $v(\pi)$ rotational motions in plane of the regular $v(\pi)$-gon that give the same $v(\pi)$-gon. 
We can represent  $G^\pi$ as follows:  $G^\pi = \{\rho^0, \rho^1 ...,\rho^{v(\pi)-1}\}$, where $\rho^i$ represents rotation of the regular $v(\pi)$-gon by the angle of  
$i \times {\frac{360^{\circ}}{v(\pi)}}$, 
where  $|G^\pi| = v(\pi)$.

Now, we are ready to define the rotational symmetry of a secondary structure and a strand ordering $\pi$, intuitively, the number of rotations of its polymer graph that give the same polymer graph, as shown in \cref{fig:sym}.

\begin{Definition}[$R$-fold rotational symmetric structure]  \label{def:sym}
	A connected unpseudoknotted secondary structure $S$ and strand ordering $\pi$ 
	(and thus polymer graph $\PolySpi$) 
	is $R$-fold rotational symmetric, 
	or simply rotationally symmetric, 
	then for any base pair $(i,j)$ in the polymer graph $\PolySpi$
	the rotation of that base pair by multiples of $(360^\circ/R)$ is also in $\PolySpi$. 
	More formally:   
	$(i_{X_k^l},j_{Y_m^n}) \in \PolySpi$,   
	iff 
	$(i_{X_k^{a(l)}}, j_{Y_m^{a(n)}}) \in \PolySpi$ for all $a \in H \! \leq \! G^\pi$, where $H$ is the largest subgroup\footnote{$H \leq G$, notationally means $H$ is a subgroup of $G$. Every subgroup of cyclic group is also cyclic.} satisfying the condition, and if $|H| = R$.  
\end{Definition}

\begin{remark}
	In Def. \ref{def:sym},	we restricted $H$ to be the {\em largest} subgroup so as to be aligned with the entropic reduction penalty due to symmetry that appears in \cref{eq:DGss}, even if  from a geometric perspective any $R$-fold rotational symmetric secondary structure should be also $R'$-fold rotational symmetric if $R'$ divides $R$. 
\end{remark}

\section{A polynomial upper bound on a class of rotationally symmetric secondary structures}\label{sec:ub}

\begin{remark}\label{rm:noNicks}
	In this section, we assume a global indexing of all bases from $1$ to $N$, and use square brackets, $[i, i+1]$, to  denote the covalent bond connecting bases $i$ and $i +1$, given that both belong to the same strand. This notation also helps to differentiate covalent bonds $[i, i+1]$ from base pair $(j,k)$ (hydrogen bonds)  notation. 
\end{remark}

\begin{Definition}[$R$-symmetric backbone cut generated by a covalent bond]\label{def:cut}
	For  a connected unpseudoknotted secondary structure $S$ and strand ordering $\pi$, the {\em $R$-symmetric backbone cut generated by the covalent bond} $b=[i_{A_m^n}, (i+1)_{A_m^n}]$ is $\mathcal{C}_R^b = \{ [i_{A_m^{a(n)}}, ({i+1})_{A_{m}^{a(n)}}]: \textrm{ for all } a \in H \!\!\leq \!\! \!\ G^\pi \textrm{ such that } |H| = R \}$. We call $b$ a {\em symmetric backbone cut generator}. An example is shown in \cref{fig:sand}.
\end{Definition}

Only one covalent bond is needed to generate its corresponding $R$-symmetric backbone cut. Also, note that this definition excludes any cut through nicks by \cref{rm:noNicks}. The next lemma shows that the number of unique symmetric backbone cuts is linear in $N$.

\subsection{Linear upper bound on number of unique symmetric backbone cuts}

\begin{lemma}[Upper bound on unique symmetric backbone cuts]\label{lem:ub}
	For any connected unpseudoknotted secondary structure $S$ of $c=\mathcal{O}(1)$ strands with $N$ total bases, with a specific strand ordering $\pi$, the number of  unique symmetric backbone cuts is $\frac{N-c}{v(\pi)} \left[ \sigma(v(\pi))-v(\pi) \right] = \mathcal{O}(N)$, where $\sigma(v(\pi))$ is sum of divisors of $v(\pi)$. 
\end{lemma}

\begin{proof}
	In general, any covalent bond is a potential candidate for a symmetric backbone cut generator. For secondary structure $S$ with $N$ bases and $c$ strands, there are $N-c$ covalent bonds  in $S$ by excluding all nicks.  
	We need to compute the total number of all possible  symmetric backbone cuts for every possible symmetry degree $R$. 
	Given a specific $R$, the number of    $R$-symmetric backbone cuts is $\frac{N-c}{R}$, 
	since for any covalent bond $x$ in a cut $\mathcal{C}_R^b$, we have $\mathcal{C}_R^x = \mathcal{C}_R^b$ (all bonds in that cut generate the same cut). 
	And, hence for any two covalent bonds $x$ and $y$, either $\mathcal{C}_R^x = \mathcal{C}_R^y$ or they are disjoint.   
	
	Because of symmetry (see \cref{lem:div}),  $R>1$ and $R$ divides $v(\pi)$ (denoted $(R\neq1)|v(\pi)$ below). 
	Assume that $d_1, d_2, ..., v(\pi)$ are divisors of $v(\pi)$ such that $d_i \neq 1$, 
	and since divisors happen in pairs ($d_i d_i' = v(\pi)$), 
	then the total number of symmetric backbone cuts is 	\begin{align*}
		\sum\limits_{(R\neq1)|v(\pi)}\frac{N-c}{R} &= 
		(N-c) \sum\limits_{(R\neq1)|v(\pi)}\frac{1}{R} \\
		&= (N-c) \left[ \frac{1}{d_1} + \frac{1}{d_2} + ... + \frac{1}{v(\pi)} \right] 
		\\
		&= (N-c) \left[ \frac{d'_1}{d_1d'_1} + \frac{d'_2}{d_2d'_2} + ... + \frac{1}{v(\pi)} \right] \\
		&= (N-c) \left[ \frac{d'_1 + d'_2 + .... + 1}{v(\pi)}\right]
		\\ 
		&=\frac{N-c}{v(\pi)} \left[ \sigma(v(\pi))-v(\pi) \right]
	\end{align*}
	which is $\mathcal{O}(N)$ since $|\pi|=c=\mathcal{O}(1)$ (i.e. number of strands $c=\mathcal{O}(1)$).
\end{proof}

If $S$ is a connected unpseudoknotted secondary structure with ordering $\pi$, you can go from any base $i$ to $j$ in two different paths around the circumference of $\PolySpi$ (clockwise or anticlockwise). We define the \emph{length} function $l[i,j]$ to be the length of the shorter path, including both $i$ and $j$ as follows: 
\begin{equation}
	l[i,j] = \min \{|i-j|+1,N-|i-j|+1\}
\end{equation} 

Also,  $\llbracket i,j \rrbracket$  is used to denote that shorter {\em segment} of length $l[i,j]$, where the direction from base $i$ to base $j$ is the same as the system strands' direction.    

\subsection{How to slice a pizza (secondary structure)}

{\bf We} want to slice any $R$-fold rotational symmetric secondary structure $S$, {\bf like pizza}, to the centre of its $\PolySpi$, without intersecting any of its base pairs. First, we formalise (\cref{def:admissible}) a special type of backbone cut, called an \emph{admissible  $R$-symmetric backbone cut}. Then, we will prove (\cref{lem:cutexist}) its existence for $S$. 

\begin{Definition}[Admissible $R$-symmetric backbone cut] \label{def:admissible}
	For any connected unpseudoknotted secondary structure $S$ with strand ordering $\pi$, the {\em $R$-symmetric backbone cut} $\mathcal{C}_R^b$ generated by $b$ is {\em admissible}, 
	if for all covalent bonds $x \in \mathcal{C}_R^b$,  
	$x$ is not ``enclosed'' by any base pair $(i,j) \in \PolySpi$, more formally:  $x \nsubseteq \llbracket i,j \rrbracket$.   An example is shown in \cref{fig:sand}.
\end{Definition}

\begin{restatable}{lemma}{cutexist} 
	\label{lem:cutexist}
	For any $R$-fold rotationally symmetric secondary structure $S$, there exists at least one admissible $R$-symmetric backbone cut of $S$.
\end{restatable}
\begin{proof}
	From $\PolySpi$, select  a base pair $(i,j)$ that has maximal length $l[i,j]$: 
	at least one of $[i-1,i]$ and $[j,j+1]$ must be a covalent bond, otherwise if both were nick $S$ would be disconnected (a contradiction). 
	We claim that this covalent bond, which we denote $[a,a+1]$, 
	is an admissible $R$-symmetric backbone cut generator, otherwise there exists a base pair $(m,n) \in \PolySpi$ such that $[a,a+1] \subseteq \llbracket m,n \rrbracket$, 
	giving a contradiction by either: 
	(a) 
	$\llbracket m,n \rrbracket$ must contain $\llbracket i,j \rrbracket$ which contradicts the maximality of $l( i,j )$, 
	or (b) $( m,n )$ and $( i,j )$ intersect forming a pseudoknot. 
	All covalent bonds in $\mathcal{C}_R^{[a,a+1]}$  have the same situation because of $R$-fold symmetry of $S$, 
	which implies that $\mathcal{C}_R^{[a,a+1]}$ is an admissible $R$-symmetric backbone cut of $S$.
\end{proof}

Note that any $R$-fold rotationally symmetric secondary structure $S$ can have more than one admissible R-symmetric cut. Before defining what we mean by a pizza slice (formally, symmetric slice in \cref{def:symmetric slice}),  the following lemma is used to  ensure such a slice is  connected.

\begin{restatable}[Pizza slicing lemma] {lemma}{connected} 
	\label{lem:connected}
	For any $R\geq 2$ and any $R$-fold rotational symmetric secondary structure $S$, let $G$ be the graph obtained from $\PolySpi$ by removing the covalent bonds of admissible $R$-symmetric backbone cut $\mathcal{C}_R^b$ generated by any covalent bond $b$, such that $G=(V(\PolySpi), E(\PolySpi)\setminus\mathcal{C}_R^b)$, then $G$ is disconnected and consists exactly of $R$ connected isomorphic components.
\end{restatable}

\begin{proof}  
	\cref{lem:cutexist} ensures the existence of at least one admissible $R$-symmetric backbone cut of~$S$, assume it is generated by $b=[i_{A_m^n}, (i+1)_{A_m^n}]$, then by \cref{def:cut}:  $\mathcal{C}_R^b = \{ [i_{A_m^{a(n)}}, ({i+1})_{A_{m}^{a(n)}}]: \forall a \in H \!\!\leq \! G^\pi, |H| = R \}$. For any two (recall, $R\geq 2$) ``consecutive'' covalent bonds  in $\mathcal{C}_R^b$ (formally:   $b_k=[i_{A_m^{a^k(n)}}, ({i+1})_{A_{m}^{a^k(n)}}]$ and  $b_{k+1}=[i_{A_m^{a^{k+1}(n)}}, ({i+1})_{A_{m}^{a^{k+1}(n)}}]$),
	we construct the ``pizza slice'' $G_k$ to be the subgraph of $\PolySpi$ induced by all vertices (or bases) that belong to $I_k = \llbracket ({i+1})_{A_{m}^{a^k(n)}},{i}_{A_{m}^{a^{k+1}(n)}} \rrbracket$. Intuitively, $G_k$ is the slice we get after cutting $\PolySpi$ at $b_k$ and $b_{k+1}$. At the end we will have a sequence of subgraphs $\mathcal{G} = \{G_1,G_2, ... , G_R\}$. Because of symmetry, all subgraphs in $\mathcal{G}$ are isomorphic. We claim that each subgraph in $\mathcal{G}$ is exactly one connected component, and $G = \bigcup_{k=1}^R G_k$. 
	
	For any $G_k$, first we will show that $G_k$ is disconnected from any other $G_{l}$ with $l \neq k$, which follows from two observations:  
	From \cref{lem:nobase} we know that the length of $I_k < \frac{N}{R}$ and $I_1, \ldots , I_R$ are disjoint segments by construction (via symmetry $R$), hence $G_k$ and  $G_l$ are not connected by a covalent bond. 
	To see that  there is no base pair $(x,y)$ connecting $G_k$ and $G_l$:  
	assume for the sake of contradiction that such a base pair $(x,y)$ exists in $\PolySpi$, then one of the two covalent bonds $b_k$ or $b_{k+1}$   $\subseteq \llbracket x,y \rrbracket$ (by the definition of $b_k, b_{k+1}$ above), contradicting the fact that $\mathcal{C}_R^b$ is an admissible backbone cut. Also, it follows directly that $G = \bigcup_{k=1}^R G_k$ from the definition of inducing subgraphs. 
	
	We next wish to show that each slice $G_k$ is connected. 
	First, there exists $d \geq 1$ such that  $G$ consists of $dR$ connected components, 
	because for each $k$, $G_k \in \mathcal{G}$ is isomorphic  
	(so if $G_k$ has $d\geq 1$ components then all $G_l \in \mathcal{G}$ do also). Next we  will show $d=1$. 
	Since $G$ is the graph obtained from the connected graph $\PolySpi$ by removing $|\mathcal{C}_R^b| = R$ covalent bonds, so there are only two cases: 
	(a) 
	if each of these $R$ covalent bonds is a cut edge, or bridge~\cite{west2001introduction}, 
	$G$ will consist of $R+1$ components, 
	contradicting the fact that number of its components must be $dR$ for some $d\geq 1$, so $d = 1$, implying that each subgraph in $\mathcal{G}$ consists of exactly one component. 
	(b) The only other case is that $G$ has $R$ components (since we've already shown that the $R$ slices are not connected to each other), giving the lemma statement. 
\end{proof}

\begin{Definition}[Symmetric slice]\label{def:symmetric slice}
	From the construction in \cref{lem:connected}, each of the $R$ 
	isomorphic subgraphs (components) in $ \mathcal{G}$ is called a {\em symmetric slice} of $\PolySpi$,  denoted  by $\rhd^{S}$. Also, the loop free energy of a symmetric slice is: 
	\begin{equation}
		\Delta G(\rhd^{S}) = \sum_{l\in \rhd^{S}} \Delta G(l)
	\end{equation}
\end{Definition}

For any $R$-fold symmetric secondary structure $S$, the following lemma shows the existence of a unique loop in the center of $\PolySpi$  surrounded by the outer base pairs of its symmetric slices. We call it the {\em central loop} of $S$, and denote it by $\bigcirc^S$. This central loop plays a crucial role in validating our slicing and swapping strategy  (\cref{lem:sand,lem:sand2}) and determining the exact upper bound (\cref{lem:polyub}) of symmetric secondary structures need to be backtracked.

\begin{restatable}{lemma}{centralloop} 
	\label{lem:centralloop}
	For any $R$-fold symmetric secondary structure $S$, there exists a single loop, that we call the central loop $\bigcirc^S$, that is not contained in any of the $R$ symmetric slices. 
	If $R>2$ then $\bigcirc^S$ is a multiloop, if $R=2$ then $\bigcirc^S$ is either a multiloop, stack, or an internal loop.   
\end{restatable}
\begin{proof}
	Using our construction in~\cref{lem:connected}, 
	let   
	$G_k \in \mathcal{G}$ be any symmetric slice, and 
	assume a local indexing from $1$ to $\frac{N}{R}$ for the bases in the segment $I_k$. Let $\mathcal{N}_k = \{(m,n) \in G_k: \neg\exists (m',n') \in G_k \textrm{ such that } m'<m<n<n'\}$. Intuitively,  $\mathcal{N}_k$ contains all outer base pairs of $G_k$, we call  $\mathcal{N}_k$ the \emph{nesting set} of $G_k$ as it determines how $G_k$ will geometrically look like (when staring at it from the centre of the pizza). 
	
	Intuitively, we construct a path $P_k$ as the concatenation of a segment of covalent bonds from $G_k$, and then a base pair in $\mathcal{N}_k$, then more covalent bonds, then a base pair, and so on.  Formally, let $d = |\mathcal{N}_k|$, then  $P_k = \llbracket 1,m_1 \rrbracket. (m_1,n_1). \llbracket n_1,m_2 \rrbracket.  (m_2,n_2) \ldots (m_c,n_c). \llbracket n_c,\frac{N}{R} \rrbracket$. Note that $\llbracket 1,m_1 \rrbracket$ and $\llbracket n_c,\frac{N}{R} \rrbracket$ may be substrands of length zero. Using this construction, $P_k$ must be a path, otherwise $G_k$ will be a disconnected graph contradicting~\cref{lem:connected}.
	
	Now, we will construct this central loop $\bigcirc^S$ as follows: for simplicity of notation, let $\mathcal{C}_R^b = \{ b, b_2, ..., b_R\}$, then 
	$\bigcirc^S = b. P_1. b_2. P_2. ... b_R. P_R$. If $R > 2$, $\bigcirc^S$ must be a multiloop, since a multiloop is defined by being bordered by $>2$ base pairs (hydrogen bonds). 
	If $R = 2$, $\bigcirc^S$ will be a multiloop if and only if $|\mathcal{N}_k| \geq 2$ (this implies there are $R |N_k| = 2|N_k|  > 2 $ base pairs bordering the central loop), 
	otherwise $\bigcirc^S$ will be an internal loop or stack. $\bigcirc^S$ can not be external loop otherwise this will contradict the connectedness of $S$, nor can $\bigcirc^S$  be a bulge nor hairpin loop as either of these  will contradict the  symmetry of $S$. 
\end{proof}

Because of the crucial importance of multiloops for our strategy, we highlight the multiloop energy model that is used in the standard dynamic programming algorithms. The free energy of a multiloop has the following linear form~\cite{dirks2007thermodynamic}:
\begin{equation} \label{eq:multi}
	\Delta G^ \textrm{multi}  = \Delta G_\textrm{init}^\textrm{multi} + b \Delta G_\textrm{bp}^\textrm{multi} + n\Delta G_\textrm{nt}^\textrm{multi}.
\end{equation}
Where, $\Delta G_\textrm{init}^\textrm{multi}$ is called the penalty for   formation of the multiloop, $\Delta G_\textrm{bp}^\textrm{multi}$ is called the penalty for each of its $b$ base pairs that border the interior of the multiloop, and $\Delta G_\textrm{nt}^\textrm{multi}$ is called the penalty for each of the $n$ free bases inside the multiloop. For any $R$-fold symmetric secondary structure~$S$, 
$b \Delta G_\textrm{bp}^\textrm{multi}$ and $n\Delta G_\textrm{nt}^\textrm{multi}$ are shared equally between the $R$ symmetric slices of $\PolySpi$, hence $R$ divides both $b$ and $n$. So, we denote  $\Delta G^ \textrm{multi}  = \Delta G_\textrm{init}^\textrm{multi} + R(\Delta G_{\rhd^{S}}^\textrm{multi})$, where $\Delta G_{\rhd^{S}}^\textrm{multi}$ is the energy contribution of each symmetric slice of $\PolySpi$ to the multiloop free energy, such that $\Delta G_{\rhd^{S}}^\textrm{multi} = \frac{b}{R} \Delta G_\textrm{bp}^\textrm{multi} + \frac{n}{R} \Delta G_\textrm{nt}^\textrm{multi}$.

\begin{note}
	For any connected unpseudoknotted secondary structure $S$ of $c$ strands, we use $\overline{\Delta G}(S)$ to denote the \snMFE of $S$, in other words $\overline{\Delta G}(S) = \sum_{l\in S} \Delta G(l)
	+  (c-1)\Delta G^{\textrm{assoc}}$. It is clear that $\overline{\Delta G}(S) \leq \Delta G(S)$, as the symmetry correction $k_\mathrm{B} T \log R \geq 0$. 
\end{note}

\begin{figure}[t]
	\centering\includegraphics[width=0.7\textwidth]{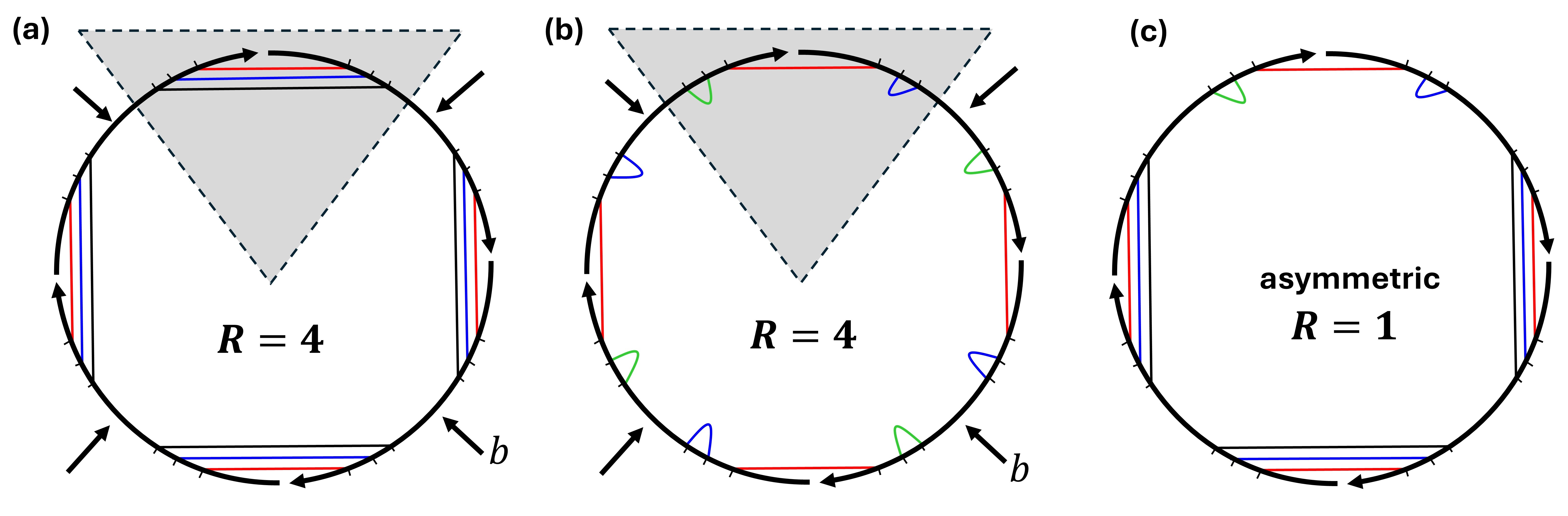}
	\caption{Slicing and swapping strategy for constructing new asymmetric structure by combining two symmetric structures with the same symmetric backbone cut. 
		(a)~4-fold symmetric secondary structure $S_i$, with admissible $4$-symmetric backbone cut $\mathcal{C}_R^b$.
		Black arrows: indicate the four covalent bonds forming $\mathcal{C}_R^b$ generated by the covalent bond $b$. 
		(b)~4-fold symmetric secondary structure $S_j$, sharing the same cut $\mathcal{C}_R^b$ as $S_i$.
		Black arrows: indicate the four covalent bonds forming $\mathcal{C}_R^b$. 
		(c)~Asymmetric secondary structure $S_k$ that is constructed by replacing the grey shaded `slice' from $S_i$ by its corresponding slice from $S_j$, using the proof of  \cref{lem:sand}. 
	}\label{fig:sand}
\end{figure}

Intuitively, the following lemma lets us take two symmetric pizzas with the same admissible symmetric cut for which we merely now their \snMFE (we don't know their true MFE), and swap a slice from one into the other to get a new asymmetric pizza whose true MFE lies between their {\snMFE}s. The key intuition is that we transforming symmetric secondary structures into an  asymmetric one. 

\begin{lemma}\label{lem:sand}[Free-energy sandwich theorem for two $R$-fold rotational symmetric  structures]
	For any two distinct $(R \geq 3)$-fold rotationally symmetric secondary structures, $S_i$ and $S_j$, of $c$ strands, such that $R \geq 3$ and $\DGnosym(S_i) \leq \DGnosym(S_j)$ and $S_i$ and $S_j$ have the same R-admissible backbone cut $\mathcal{C}_R^b$, then there exists at least one asymmetric secondary structure $S_k$, such that $\DGnosym(S_i) \leq \Delta G(S_k) \leq \DGnosym(S_j)$.  Furthermore, the statement holds if $R=2$ and at least one of the central loops  $\bigcirc^{S_i},  \bigcirc^{S_j}$  is a multiloop. 
\end{lemma}

\begin{proof}
	If $R \geq 3$,  from \cref{lem:centralloop},  there exist two unique central multiloops for $S_i$ and $S_j$, denoted  by $\bigcirc^{S_i}$ and $\bigcirc^{S_j}$. Also, if $R=2$, by hypothesis, $\bigcirc^{S_i}$ and $\bigcirc^{S_j}$ are multiloops. 
	We prove the claim with two cases:
	\begin{enumerate}
		\item[Case 1.]  $\DGnosym(S_i) = \DGnosym(S_j)$: 
		\begin{gather}
			\sum_{l\in S_i} \Delta G(l)
			+  (c-1)\Delta G^{\textrm{assoc}} = \sum_{l\in S_j} \Delta G(l)
			+  (c-1)\Delta G^{\textrm{assoc}} \label{eq:sandwich1}
			\\
			\sum_{l\in S_i} \Delta G(l)
			= \sum_{l\in S_j} \Delta G(l)
			\\
			R(\Delta G(\rhd^{S_i}))+ \Delta G(\bigcirc^{S_i})
			= R(\Delta G(\rhd^{S_j}))+ \Delta G(\bigcirc^{S_j})
			\\
			R(\Delta G(\rhd^{S_i})) + R \Delta G_{\rhd^{S_i}}^\textrm{multi} + \Delta G^\textrm{multi}_\textrm{init}
			= 	R(\Delta G(\rhd^{S_j})) + R \Delta G_{\rhd^{S_j}}^\textrm{multi} +  \Delta G^\textrm{multi}_\textrm{init}
			\\
			R(\Delta G(\rhd^{S_i}) + \Delta G_{\rhd^{S_i}}^\textrm{multi} ) = R(\Delta G(\rhd^{S_j}) + \Delta G_{\rhd^{S_j}}^\textrm{multi})
			\\
			\Delta G(\rhd^{S_i}) + \Delta G_{\rhd^{S_i}}^\textrm{multi}  = \Delta G(\rhd^{S_j}) + \Delta G_{\rhd^{S_j}}^\textrm{multi} \label{eq:sandwich2}
		\end{gather}
		We define a new secondary structure $S_k$ using our slicing and swapping strategy, shown in \cref{fig:sand}, by removing one slice from $S_i$, and adding its corresponding slice from~$S_j$. 
		\cref{lem:connected} guarantees that the new secondary structure  $S_k$ is connected and unpseudoknotted. Furthermore,  $S_k$ is asymmetric since $S_i \neq S_j$. 
		$S_k$ being asymmetric means that it's rotational symmetry $R_k=1$, $k_B T \log R_k = 0$ hence we can write: 
		\begin{gather*}
			\Delta G(S_k) = (R-1)\left( \Delta G(\rhd^{S_i}) + \Delta G_{\rhd^{S_i}}^\textrm{multi}) \right) + 
			G^\textrm{multi}_\textrm{init} + (c-1)\Delta G^{\textrm{assoc}} + \left(	\Delta G(\rhd^{S_j}) +\Delta G_{\rhd^{S_j}}^\textrm{multi} \right)
			\\
			\Delta G(S_k) = R\left(\Delta G(\rhd^{S_i}) +  \Delta G_{\rhd^{S_i}}^\textrm{multi} \right) + \Delta G^\textrm{multi}_\textrm{init} + (c-1)\Delta G^{\textrm{assoc}}
		\end{gather*}
		with the final step uses \cref{eq:sandwich2}. 
		Hence  $\DGnosym(S_i) = \Delta G(S_k) = \DGnosym(S_j)$. 
		
		\item[Case 2.]  $\DGnosym(S_i) < \DGnosym(S_j)$:  Following the same algebraic manipulation as \cref{eq:sandwich1} to \cref{eq:sandwich2}, but replaceing $=$ with $<$, we get the following:
		\begin{gather*}
			\Delta G(\rhd^{S_i}) + \Delta G_{\rhd^{S_i}}^\textrm{multi}  < \Delta G(\rhd^{S_j}) + \Delta G_{\rhd^{S_j}}^\textrm{multi}
		\end{gather*}
		As before, we define a new connected asymmetric secondary structure $S_k$, using the same slicing and swapping strategy, resulting in: $\DGnosym(S_i) < \Delta G(S_k) < \DGnosym(S_j)$. 
	\end{enumerate}
\end{proof}

\cref{lem:sand} states that, if two symmetric secondary structures, having the same admissible R-symmetric backbone cut, belong to the same energy level based on \symnMFE algorithm, ignoring symmetry entropic correction, this implies the existence of at least one asymmetric secondary structure that actually belong to the same energy level because symmetry  correction for asymmetric structures is zero. If the two secondary structures belong to two different energy levels, then there exist at least one asymmetric secondary structure that actually belong to an energy level strictly lies between the other two energy levels.

\paragraph{Intuition for the case of $R=2$ and the central loop is not a multiloop.} 
When $R=2$, and the central loop is not a multiloop, the proof of \cref{lem:sand} breaks.   
From \cref{lem:centralloop}, when $R=2$ the central loop is either a multiloop, internal loop or stack loop. 
The multiloop case has been handled already (\cref{lem:sand}), and the stack loop case can be subsumed into the  internal loop case, since stacks are considered a special type of internal loop in the standard energy model~\cite{dirks2007thermodynamic}. 
Instead of depending on having the same admissible 2-symmetric backbone cut, we depend on sharing the same central internal loop itself, this more restricted hypothesis   implies having the same admissible 2-symmetric backbone cut too, allowing us to prove \cref{lem:sand2} using a similar strategy to  \cref{lem:sand}.

\begin{lemma}\label{lem:sand2}[Free-energy sandwich theorem for two $2$-fold rotational symmetric structures]
	For any two distinct  $2$-fold rotationally symmetric secondary structures, $S_i$ and $S_j$, of $c$ strands, such that $\DGnosym(S_i) \leq \DGnosym(S_j)$ and both have the same central internal loop $\bigcirc^{S_i} = \bigcirc^{S_j}$, then there exists at least one asymmetric secondary structure $S_k$, such that $\DGnosym(S_i) \leq \Delta G(S_k) \leq \DGnosym(S_j)$.  
\end{lemma}
\begin{proof}
	Since $\bigcirc^{S_i} = \bigcirc^{S_j}$,  then both have the same admissible 2-symmetric backbone cut as any covalent bond $b$ that belong to any side of the internal loop can be its generator.
	\begin{gather*}
		\DGnosym(S_i) \leq \DGnosym(S_j)
		\\
		\sum_{l\in S_i} \Delta G(l)
		+  (c-1)\Delta G^{\textrm{assoc}} \leq \sum_{l\in S_j} \Delta G(l)
		+  (c-1)\Delta G^{\textrm{assoc}}
		\\
		\sum_{l\in S_i} \Delta G(l)
		\leq \sum_{l\in S_j} \Delta G(l)
		\\
		2(\Delta G(\rhd^{S_i}))+ \Delta G(\bigcirc^{S_i})
		\leq 2(\Delta G(\rhd^{S_j}))+ \Delta G(\bigcirc^{S_j})
		\\
		\Delta G(\rhd^{S_i})
		\leq \Delta G(\rhd^{S_j})
	\end{gather*}	
	We define a new secondary structure $S_k$ using our slicing and swapping strategy, shown in \cref{fig:sand}, by removing one slice  (half in this case) from $S_i$ , and adding its corresponding slice from~$S_j$. 
	Note that $S_k$ will have also the same central internal loop $\bigcirc^{S_k}= \bigcirc^{S_i}$. \cref{lem:connected} guarantees that $S_k$ is connected and unpseudoknotted. Since $S_k$ is asymmetric: 
	\begin{gather*}
		\Delta G(S_k) = \Delta G(\rhd^{S_i}) +\Delta G(\rhd^{S_j}) + \Delta G(\bigcirc^{S_k}) + (c-1)\Delta G^{\textrm{assoc}}
	\end{gather*}
	
	Which implies that $\DGnosym(S_i) \leq \Delta G(S_k) \leq \DGnosym(S_j)$. 
\end{proof}

We now have two sandwich theorems that we can use to construct an asymmetric structure: \cref{lem:sand,lem:sand2}. In \cref{sec:BT} we give a backtracking algorithm to search for suitable $S_i$ and $S_j$, with the goal of  applying either one of these two sandwich theorems to $S_i$ and $S_j$. 
To get an overall polynomial bound for the backtracking algorithm, we wish to bound, given $S_i$, how many secondary structures to scan before finding a suitable $S_j$.    
\cref{lem:ub} gives this upper bound on this number  when applying \cref{lem:sand}. 
Next, \cref{lem:ub2} gives this upper bound  when applying  \cref{lem:sand2}. 

Unfortunately, the bound in \cref{lem:ub2} is larger than \cref{lem:ub}, since the energy model is more complex for internal loops than multiloops~\cite{dirks2003partition}. 

\begin{lemma}[Upper bound on number of  central internal loops]\label{lem:ub2}
	For any set of $c$ strands with specific ordering $\pi$, for any set $\mathcal{T}$ of $2$-fold rotational symmetric secondary structures $(R=2)$, 
	such that each has a distinct internal central loop, the  cardinality of $\mathcal{T}$, $ |\mathcal{T}| \leq \sum_{s\in y} (\lVert \baseA \rVert_s \lVert \baseT \rVert_s+ \lVert \baseG \rVert_s \lVert \baseC \rVert_s ) \leq N^2/16$, 
	where $y$ is a fundamental component of $\pi$, such that $\pi = y^2$, and $\parallel\! \! B \!\! \parallel_s$  
	denotes the number of bases in strand $s$ of type $B$ for all $B\in\{\mathrm{A},\mathrm{T},\mathrm{G},\mathrm{C}\}$.
\end{lemma} 

\begin{proof}
	Let  $\mathcal{T}$ be any set of of $2$-fold rotational symmetric secondary structures.  
	Since each $S \in \mathcal{T}$ has a distinct internal central loop, we  focus only on giving an upper bound on the maximum number of distinct internal central loops. 
	Since $R=2$, the two base pairs forming any  central internal loop must be within two strands of the same type $X$  of the same order $m$ within their fundamental components (for example the strands are $X_m^1$ and $X_m^2$), otherwise we  have a disconnected secondary structure due to the existence of nicks. 
	
	Only one of the two base pairs of any internal central loop needs to be specified explicitly, since, by symmetry, the other base pair is   automatically determined. 
	So, by considering all base pairs (including many that are irrelevant), the number of all distinct central internal loops is  $\leq \sum_{s\in y} (\parallel\! \! \baseA \!\! \parallel_s \parallel\! \! \baseT \!\! \parallel_s+ \parallel\! \! \baseG \!\! \parallel_s\parallel\! \! \baseC \!\! \parallel_s)$. 
	Hence, $ |\mathcal{T}| \leq \sum_{s\in y} (\parallel\! \! \baseA \!\! \parallel_s \parallel\! \! \baseT \!\! \parallel_s+ \parallel\! \! \baseG \!\! \parallel_s\parallel\! \! \baseC \!\! \parallel_s)$.
	Using the two following number theoretic facts:      
	\begin{itemize} \item
		If we have two indistinguishable strands, of the same type, $X$, of length $n$, the maximum intra-base pairs between them happens when the sequence of $X$ is a word over $\{\baseA,\baseT\}$ or $\{\baseG,\baseC\}$ such that $\parallel\! \! A \!\! \parallel_X = \lfloor \frac{n}{2}\rfloor$ and $\parallel\! \! T \!\! \parallel_X = \lceil \frac{n}{2}\rceil$ or vice versa, and the same for $\{\baseG,\baseT\}$.    
	\end{itemize}
	\begin{itemize} \item
		For any integer $n>0$, if $n = n_1 + n_2+ ... + n_k$, such that $n_i\geq 0$ for all $i \in \{1,2, \ldots k\}$, then $n^2 \geq n_1^2 + n_2^2+ ... + n_k^2$.
	\end{itemize}
	We get the following: 
	\begin{align*}
		\sum_{s\in y} (\parallel\! \! \baseA \!\! \parallel_s \parallel\! \! \baseT \!\! \parallel_s+ \parallel\! \! \baseG \!\! \parallel_s\parallel\! \! \baseC \!\! \parallel_s) & \leq  
		\left\lceil\dfrac{|s_1|}{2}\right\rceil
		\left\lfloor\dfrac{|s_1|}{2}\right\rfloor + 
		\left\lceil\dfrac{|s_2|}{2}\right\rceil
		\left\lfloor\dfrac{|s_2|}{2}\right\rfloor + \ldots + 
		\left\lceil\dfrac{|s_{c/2}|}{2}\right\rceil
		\left\lfloor\dfrac{|s_{c/2}|}{2}\right\rfloor \\
		& \leq 
		\left(\frac{|s_1|}{2}\right)^2 + \left(\frac{|s_2|}{2}\right)^2 + \ldots + \left(\frac{|s_{c/2}|}{2}\right)^2 \\
		&= 
		\frac{|s_1|^2 + |s_2|^2 + \ldots + |s_{c/2}|^2}{4} \\
		& \leq 
		\frac{(N/2)^2}{4} 
		=
		\frac{N^2}{16} 
	\end{align*}
	where $s_i, i \in \{1,\ldots,c\}$, are strand types. 
	Hence:  	 $ |\mathcal{T}| \leq \sum_{s\in y} (\parallel\! \! \baseA \!\! \parallel_s \parallel\! \! \baseT \!\! \parallel_s+ \parallel\! \! \baseG \!\! \parallel_s\parallel\! \! \baseC \!\! \parallel_s ) \leq N^2/16$. 
\end{proof}

\subsection{Polynomial upper bound on number of symmetric secondary structures (for future backtracking)}
\begin{lemma}\label{lem:polyub}
	
	Given an ordering $\pi$ of $c$ strands, 
	for any set $\mathcal{T}$ of distinct symmetric secondary structures such that 
	\begin{enumerate}
		\item  for any two $(R>2)$-fold symmetric secondary structures $S_i, S_j \in \mathcal{T}$, where $S_i$ and $S_j$ have different admissible R-symmetric backbone cuts (we mean all possible cuts are different), 	
		and 
		\item 
		for any two 2-fold symmetric secondary structures $S_i, S_j \in \mathcal{T}$, where $S_i$ and $S_j$ have different admissible R-symmetric backbone cuts (all possible cuts are different)
		or  different central internal loops, 
	\end{enumerate}
	then $|\mathcal{T}| \leq \mathcal{U}$, where $\mathcal{U} =  \frac{N-c}{v(\pi)} \left[ \sigma(v(\pi))-v(\pi) \right] + \frac{N^2}{16} = \mathcal{O}(N^2)$. 
\end{lemma}

\begin{proof}
	The proof is a trivial  implication of \cref{lem:ub,lem:ub2}.  
	(More formally: Let $\mathcal{U}= \mathcal{U}_1+\mathcal{U}_2$ where 
	$\mathcal{U}_1$ is the upper bound on the number of on unique symmetric backbone cuts (\cref{lem:ub}),
	and
	$\mathcal{U}_2$ is the upper bound on the number of unique central internal loops (\cref{lem:ub2}). 
	Assume for the sake of contradiction that $ |\mathcal{T}| > \mathcal{U}$. 
	From the pigeon hole principle, $|\mathcal{T}| > \mathcal{U}$ implies repeating at least one symmetric backbone cut or a central internal loop contradicting the hypothesis about structures of $\mathcal{T}$.)
\end{proof}

The (bad) quadratic bound in \cref{lem:ub2} is not that frequent: 
In particular, that bound only appears when $R=2$ and the central loop is an internal loop for both symmetric secondary structures (since $R=2$ this implies that the repetition number for every strand type is {\em even}, which in practice, say, for random or typical systems, may not be frequent).   
In particular the following lemma gives a {\em linear} bound when the repetition number of at least one strand type is odd.

\begin{lemma} \label{lem:even}
	For any $R$-fold rotationally symmetric  secondary structure $S$, with ordering $\pi$, such that $R$ is even, then the repetition number of each strand type must be even. Hence for any system of $c$ strands ($k$ strand types) such that the repetition number of some strand type is odd, then $\mathcal{U}$, where $\mathcal{U} =  \frac{N-c}{v(\pi)} \left[ \sigma(v(\pi))-v(\pi) \right] = \mathcal{O}(N)$. 
\end{lemma}
\begin{proof}
	Suppose $S$ is a $R$-fold rotational symmetric such that $R$ is even. From \cref{lem:div}, $R$ divides $v(\pi)$, which implies $v(\pi)$ is even too. Since $\pi = x^{v(\pi)}$ then $\pi = y^2$ such that $y= x^{v(\pi)/2}$, and this is valid only if the repetition number of each strand type is even. The linearity of $\mathcal{U}$ follows directly if repetition number of at least one strand type is odd.
\end{proof}


\section{Backtracking to find the true MFE}\label{sec:BT}

In this section, we give a backtracking procedure, 
\cref{algo:2} in \cref{app:backalgo}, 
to give our main result (\cref{thm:main}). 
First,  we run our augmentation of the known \symnMFE (\snMFE) algorithm---\cref{algo:1} in \cref{sec:AlgoMFE}, which returns some matrices which are input to the backtracking algorithm, \cref{algo:2}. 
Our multistranded backtracking algorithm builds on the single-stranded backtracking algorithm of Wuchty et al.~\cite{wuchty1999complete},
which in turn follows Waterman and Byers~\cite{waterman1985dynamic}, 
 although we make several technical modifications. 
In particular, distinctions with that previous work~\cite{wuchty1999complete,waterman1985dynamic} include: 
\begin{enumerate}
	\item We generalise backtracking from single-stranded to multistranded, which has a slightly different MFE algorithm, consistent with Dirks et al and Fornace et al~\cite{dirks2007thermodynamic,fornace2020unified} (i.e.~as implemented by the NUPACK software); in particular to ensure  the connectedness of secondary structures (a non-issue for~\cite{wuchty1999complete,waterman1985dynamic}). 
	\item We make major changes to the core of the backtracking algorithm  to ensure generation of all secondary structures, 
	at each of a specified number of energy levels, in energy level order 
	(which is different from the Wuchty et al's~\cite{wuchty1999complete} approach of backtracking all sub-optimal secondary structures that lie between the \snMFE and any arbitrary upper limit above it). 
	\item We extend the refinement cases of  Wuchty et al~\cite{wuchty1999complete} to handle our auxiliary matrices in such a way that yields  a good running time. 
\end{enumerate}

\subsection{Partially and fully specified structures}

\begin{Definition}[Partially and fully specified structure $\mathcal{S}$]
	A partially specified structure $\mathcal{S} = (\delta,\mathcal{P},E_{L_\mathcal{S}})$, where $\delta$ is a stack of disjoint segments of one or more DNA strands $\{[i,j]^t. [k,l]^{t'} \ldots\}$ where $[i,j]^t$ is the top of the stack, such that $i$ and $j$ are the end bases of the segment $[i,j]$, $t \in \{\square,b,m\}$ is the type of each segment, 
	such that $t = \square$ means the existence of a base pair between $i$ and $j$, is as yet undetermined, 
	$t =b$ means there is a base pair between $i$ and $j$, and 
	$t = m$ means that entire segment $[i,j]$ is part of a multiloop. 
	$\mathcal{P}$ is a set of base pairs formed in $\mathcal{S}$, and $E_{L_\mathcal{S}}$ is the energy of all loops that are 
	`completely formed' in $\mathcal{S}$. If $\delta = \phi$, we call $\mathcal{S}$ a fully specified structure.
\end{Definition}

A fully specified structure is a connected unpseudoknotted secondary structure. For any segment $[i,j]^t$,  label $t$ is assigned according to how a segment is generated through \emph{refinement} from another segment, formalized in \cref{sec:backhigh}, 
label $t$ is needed in switching the backtracking between the appropriate matrices of the \snMFE algorithm.
We will denote the minimum free energy attainable from segment $[i,j]^t$, by $E([i,j]^t)$, which we get directly from the appropriate matrix $M$, $M^b$, $M^m$, $M^\text{b:int}$, $M^\text{b:mul}$, or $M^\text{m:2}$, that are returned  by the \snMFE algorithm (\cref{algo:1}), based on $t$, full details are in \cref{sec:AlgoMFE}. 
The domain of the  function $E$ is extended to include the set of all partially specified structures, in addition to the set of all segments,   
$\mathcal{S} = (\delta,\mathcal{P},E_{L_\mathcal{S}})$ so that 
$E$ gives  the minimum free energy attainable from $\mathcal{S}$, 
respecting the refinement rules formalized in \cref{sec:backhigh}, as follows:
\begin{equation}\label{eq:ES}
	E(\mathcal{S}) = E_{L_{\mathcal{S}}} 
	+ \sum \limits_{[m,n]^t \in \delta} E([m,n]^t)
\end{equation}

Any partially specified structure $\mathcal{S} = (\delta,\mathcal{P},E_{L_\mathcal{S}})$ represents a set of all structures that have the base pairs $\mathcal{P}$ in common:  we can think of $\mathcal{S}$ as the root of the tree of these structures, all intermediate nodes of this tree will be partially specified structures, and its leaves will be fully specified structures, and $E(\mathcal{S})$ is the minimum free energy attainable from this tree where all its nodes, structures, are further \emph{refinements} of $\mathcal{S}$.   

\begin{Definition}[Refinement of a partially specified structure]
	A structure $\mathcal{S}' = (\delta',\mathcal{P}',E_{L_{\mathcal{S}'}})$ is called a refinement of the partially specified structure $\mathcal{S} = (\delta,\mathcal{P},E_{L_\mathcal{S}})$ if $\mathcal{P} \subseteq \mathcal{P}'$, and for each segment $[i',j']^{t'} \in \delta'$ there exist a segment $[i,j]^t \in \delta$ such that $[i',j']^{t'} \subseteq [i,j]^t$.
\end{Definition}

\subsection{Analysis of the backtracking algorithm refinement rules}\label{sec:backhigh}

The backtracking algorithm  starts with $\mathcal{S} = ([1,N]^\square, \phi,0)$, which represents the whole system of strands with a specific strand order, $\pi$, without any base pair formed, $\mathcal{P} = \phi$, hence no loops are formed too, $E_{L_\mathcal{S}} = 0$. $\mathcal{S}$ is the parent node of the tree of any possible structure. 
Now, we will outline the main refinement procedure of the generic partial structure $\mathcal{S} = ([i,j]^t.\delta, \mathcal{P}, E_{L_{\mathcal{S}}})$ that has been chosen, at the beginning of each iteration of the backtracking algorithm ({\cref{algo:2}}), 
from some array $\mathcal{R}_{z}$, the array of partially specified structures associated with each secondary structure $S_{z} \in \{S_1, \ldots, S_\mathcal{U}\}$, which is the secondary structure number $z$, of the worst case $\mathcal{U}$ secondary structures we need to scan (\cref{lem:polyub}), such that $S_z$ is completely scanned during the backtracking. 
The segment $I = [i,j]^t$, the top of the
segments stack of $\mathcal{S}$, will be popped  and  refined based on the type of  label $t$ resulting in a new refined structure $\mathcal{S}'$. Matrices $M$, $M^b$, $M^m$, and the new auxiliary matrices $M^\text{b:int}$, $M^\text{b:mul}$, and $M^\text{m:2}$, returned by  \cref{algo:1}, will be used to compute the minimum free energy $E(\mathcal{S}')$ attainable from  the refined partially specified structure~$\mathcal{S}'$.

Given that the algorithm  scans, or backtracks, all secondary structures in energy level $\mathcal{E}$, and $\mathcal{B}$ is the best candidate for the true MFE at the moment, then the \emph{acceptance criteria} of any refined partially specified structure $\mathcal{S}'$ is: 
\begin{equation}
	E(\mathcal{S}') \leq    \mathcal{B}
\end{equation}

This acceptance criteria is checked after each refinement case, and if it is satisfied, $\mathcal{S}'$ will be added on the array of partially specified structures $\mathcal{R}_u$, for some $u \in \{1, \ldots, \mathcal{U} \}$, for further refinements in future, where $\mathcal{R}_u$ is the secondary structure currently being  scanned. 
Note that, because of the strict sequential scanning of the backtracking algorithm (\cref{remark:newS}), 
the acceptance criteria implicitly implies that $\mathcal{E} \leq E(\mathcal{S}')$. Also, 
the acceptance criteria guarantees the connectedness of at least one potential fully specified structure which is a child of $\mathcal{S}'$  
(in  \cref{algo:2}, setting $E(\mathcal{S}') = +\infty$ implies a disconnected or invalid structure). 

There are $3$ cases based on the type of $t$ of the segment $I = [i,j]^t$ that has been popped (as mentioned above) from the stack $\delta$:

\begin{enumerate}\item
	$t = \square$ (recall, $\square$ means: the existence of a base pair between $i$ and $j$ is undetermined): 
	
	In this case we backtrack in matrix $M$.  
	
	$i$ and $j$ are the end bases of $I$, and the possible refinements, based on Eq.~(1) in \cref{fig:mfe}, are: 
	
	\begin{itemize}
		\item Subcase: If the base $j$, at the $3'$ end of the segment $[i,j]$, is unpaired, that will result in the new partial structure (i.e.~excluding $j$ and moving to $j-1$):
		
		$\mathcal{S}' = ([i,j-1]^\square.\delta, \mathcal{P}, E_{L_{\mathcal{S}}})$ such that $E(\mathcal{S}') = M_{i,j-1} + E_{L_{\mathcal{S}}} 
		+ \sum \limits_{[m,n]^t \in \delta} E([m,n]^t)$.
		
		\item 	Subcase: If the base $j$ forms a base pair with base $d \in [i,j-1]$,  we need to scan all such  $d \in [i,j-1]$, 
		and for each we have the new partial structure:
		
		$\mathcal{S}' = ([i,d-1]^\square.[d,j]^b.\delta, \mathcal{P}, E_{L_{\mathcal{S}}})$ such that $E(\mathcal{S}') = M_{i,d-1} + M^b_{d,j}  +E_{L_{\mathcal{S}}} 
		+ \sum \limits_{[m,n]^t \in \delta} E([m,n]^t)$. 
		
		Note that we did not add the base pair $(d,j)$ to $\mathcal{P}$ at this step, but we shall do when refining the interval
		$[d,j]^b$ enclosed by $(d,j)$~\cite{wuchty1999complete}. 	
	\end{itemize}
	
	Then the acceptance criteria will be checked after each of the two sub-cases above. The backtracking algorithm have to  check up to $\mathcal{O}(N)$ refined structures (because of $d$ spans $\leq N$ bases in subcase 2), and hence save up to  $\mathcal{O}(N)$ refined structures to $\mathcal{R}_u$ in the worst case.   
	
	\item
	Case $t = b$ (recall: a base pair is formed between the end bases $i$ and $j$ of $[i,j]^b$ segment):
	
	In this case we backtrack in matrix $M^b$. 	
	Now, assume the segment $[i,j]^b$ is popped from the stack $\delta$, based on Eq.~(2) in \cref{fig:mfe}, there are four subcases:
	
	\begin{itemize}
		
		\item \textbf{Hairpin loop} formation: If  $(i,j)$ is closing a hairpin loop (\cref{fig:mfe}(b)), this will result in the new partial structure:
		
		$\mathcal{S}' = (\delta, \mathcal{P} \cup \{(i,j)\}, E_{L_{\mathcal{S}}} + \Delta G_{i,j}^\text{hairpin} )$ such that $E(\mathcal{S}') = \Delta G_{i,j}^\text{hairpin} + E_{L_{\mathcal{S}}} 
		+ \sum \limits_{[m,n]^t \in \delta} E([m,n]^t)$.

		\item \textbf{Interior loop} formation: We need to scan all possible base pairs $(d,e)$ that bind to form an interior loop along with $(i,j)$ (\cref{fig:mfe}(b)). 
		Scanning all pairs in a straightforward way would lead to checking up to $\mathcal{O}(N^2)$ refined structures, 
		which would end up with a poor final worse-case time complexity of the backtracking algorithm. 
		Instead, we scan these base pairs in a different way, that keeps the number of refined structures that we need to check at each iteration to $\mathcal{O}(N)$. We achieve this by introducing a new auxiliary matrix, called $M^\text{b:int}$, in the \snMFE algorithm,  \cref{algo:1}. 		
		
		Modifying the generic form of segment $I$ is essential, in this case, the new segment generic form will be $I = [i,j]^b_{\text{int}:k}$, such that $k \in [i+1,j-1]$, which means that any base $d \in[k,j-1]$ is unpaired, hence not included in the second base pair formation to complete the interior loop with $(i,j)$.
		When a new segment $[i,j]^b$ is just generated as a refinement from another segment, 
		$[i,j]^b$ will be interpreted inside this case as $ [i,j]^b_{\text{int}:j} $, which means any base $d \in [i+1,j-1]$ can be part of the second base pair closing the current interior loop. 
		Given  $I = [i,j]^b_{\text{int}:k}$, there are two cases:
		
		1) If the base $k-1$ is also unpaired, this will
		result in the new partial structure: 
		
		$\mathcal{S}' = ([i,j]^b_{\text{int}:k-1}.\delta, \mathcal{P}, E_{L_{\mathcal{S}}})$ such that $E(\mathcal{S}') = M_{i,j,k-2}^\text{b:int} + E_{L_{\mathcal{S}}} 
		+ \sum \limits_{[m,n]^t \in \delta} E([m,n]^t)$.
		
		2) If the base $k-1$ is paired with another base $d \in [i+1,k-2]$ closing the interior loop, this will result in the new partial structure: 
		
		$\mathcal{S}' = ([d,k-1]^b.\delta, \mathcal{P} \cup \{(i,j)\}, E_{L_{\mathcal{S}}} + \Delta G_{i,d,k-1,j}^\text{interior})$ such that $E(\mathcal{S}') = M^b_{d,k-1} + \Delta G_{i,d,k-1,j}^\text{interior}   +E_{L_{\mathcal{S}}} 
		+ \sum \limits_{[m,n]^t \in \delta} E([m,n]^t)$. 
		
		\item \textbf{Multiloop} formation: In this case we also need to scan all possible pairs $(d,e)$ that will used to form a multi-loop to the $3'$ end (\cref{fig:mfe}(b)), so we will follow the same strategy as the case of interior loop formation, by introducing another new auxiliary matrix, called $M^\text{b:mul}$, in the symmetry agnostic MFE algorithm,  \cref{algo:1}. 
		The generic form of $I$ in this case will be updated to $I = [i,j]^b_{\text{mul}:k}$ such that $k \in [i+1,j-1]$, which means that any base $d \in[k,j-1]$ is unpaired, hence not included in the forming any base pair inside this multiloop along with $(i,j)$.
		When a new segment $[i,j]^b$ is just generated as a refinement from another segment, 
		$[i,j]^b$ will be interpreted inside this case as $ [i,j]^b_{\text{mul}:j} $, which means any base $d \in [i+1,j-1]$ can be part of base pair formation inside this multiloop. 
		Given  $I = [i,j]^b_{\text{mul}:k}$, there are two cases:
		
		1) If $k-1$ is also unpaired, this will
		result in the new partial structure:

		$\mathcal{S}' = ([i,j]^b_{\text{mul}:k-1}.\delta, \mathcal{P}, E_{L_{\mathcal{S}}})$ such that $E(\mathcal{S}') = M_{i,j,k-2}^\text{b:mul} + E_{L_{\mathcal{S}}} + \sum \limits_{[m,n]^t \in \delta} E([m,n]^t)$. 
		
		2) If the base $k-1$ is paired with another base $d \in [i+1,k-2]$, this will result in the new partial structure: 
		
		$\mathcal{S}' = ([i+1,d-1]^m.[d,k-1]^b.\delta, \mathcal{P} \cup \{(i,j)\}, \Delta G_\text{init}^\text{multi} + 2\Delta G_\text{bp}^\text{multi} + (j-k) \Delta G_\text{nt}^\text{multi} + E_{L_{\mathcal{S}}})$ such that $E(\mathcal{S}') =  M^m_{i+1,d-1} + M^b_{d,k-1} + \Delta G_\text{init}^\text{multi} + 2\Delta G_\text{bp}^\text{multi} + (j-k)\Delta G_\text{nt}^\text{multi} +E_{L_{\mathcal{S}}} 
		+ \sum \limits_{[m,n]^t \in \delta} E([m,n]^t)$. 
		
		\item \textbf{Exterior loop} formation: All bases $z \in [i,j]$ such that $[z,z+1]$ is a nick \cref{fig:sec struct}, transition between two strands, are scanned (\cref{fig:mfe}(b)), leading to the new partial structure: 
		
		$\mathcal{S}' = ([i+1,z]^\square.[z+1,j-1]^\square.\delta, \mathcal{P} \cup \{(i,j)\}, E_{L_{\mathcal{S}}})$ such that $E(\mathcal{S}') = M_{i+1,z} +  M_{z+1,j-1} +  E_{L_{\mathcal{S}}} + \sum \limits_{[m,n]^t \in \delta} E([m,n]^t)$.

	\end{itemize}
	
	Then the acceptance criteria will be checked after each sub-case. Now, with the aid of the introduced new auxiliary matrices $M^\text{b:int}$ and $M^\text{b:mul}$, the backtracking algorithm checks up to $\mathcal{O}(N)$ refined structures, and hence saves up to  $\mathcal{O}(N)$ refined structures to $\mathcal{R}_u$ in the worst case. Without the  new auxiliary matrices (in this case of $t=b$), the backtracking algorithm will check up to $\mathcal{O}(N^2)$ refined structures, and saves up to  $\mathcal{O}(N^2)$ refined structures to $\mathcal{R}_u$.

	\item
	Case $t = m$ 	(recall: the segment $[i,j]^m$ is a part of a multiloop): 
	
	We backtrack in matrix $M^m$. 
	Now, assume the segment $[i,j]^m$ is popped from the stack, based on Eq.~(3) in \cref{fig:mfe},  there are two subcases:
	
	\begin{itemize}
		\item  If there exists exactly one additional base pair $(d,e)$ defining the multiloop (\cref{fig:mfe}(c)), then we will scan for all possible pairs $(d,e)$ that could be used. Following the same strategy of scanning introduced before (i.e.~in case 2. $t=b$, interior loop or multiloop) to reduce time in interior and multiloop formation cases when $t = b$, there are two cases: 
		
		1) If $j$ is unpaired, this will
		result in the new partial structure: 
		
		$\mathcal{S}' = ([i,j-1]^m.\delta, \mathcal{P}, E_{L_{\mathcal{S}}} + \Delta G_\text{nt}^\text{multi})$ such that $E(\mathcal{S}') = M_{i,j-1}^m + \Delta G_\text{nt}^\text{multi} +E_{L_{\mathcal{S}}} + \sum \limits_{[m,n]^t \in \delta} E([m,n]^t)$. 	 	
		
		2) If the base $j$ is paired with another base $d \in [i,j-1]$ defining the multiloop, this will result in the new partial structure: 
		
		$\mathcal{S}' = ([d,j]^b.\delta, \mathcal{P}, E_{L_{\mathcal{S}}} + \Delta G_\text{bp}^\text{multi} + (d-i) \Delta G_\text{nt}^\text{multi})$ such that $E(\mathcal{S}') = M^b_{d,j} + \Delta G_\text{bp}^\text{multi}+ (d-i) \Delta G_\text{nt}^\text{multi} +E_{L_{\mathcal{S}}} 
		+ \sum \limits_{[m,n]^t \in \delta} E([m,n]^t)$.

		\item  If there exist more than one  base pair defining the multiloop, all possible pairs $(d,e)$ are scanned (\cref{fig:mfe}(c)). Following the same strategy of scanning, and using one of the new auxiliary matrices, called $M^\text{m:2}$, in the \snMFE algorithm, \cref{algo:1}.

		The generic form of $I$ in this case will be updated to $I = [i,j]^m_{\text{mul}:k}$ 	such that $k \in [i,j]$, which means that any base $d \in[k,j]$ is unpaired, hence not included in the forming of any base pair inside this multiloop.
		When a new segment $[i,j]^m$ is just generated as a refinement from another segment, 
		$[i,j]^m$ will be interpreted inside this case as $ [i,j]^m_{\text{mul}:j-1} $, which means any base $d \in [i,j]$ can be part of base pair formation inside this multiloop. 
		Given  $I = [i,j]^m_{\text{mul}:k}$, there are two cases:

		1) If the base $k-1$ is also unpaired, this will
		result in the new partial structure: 
		
		$\mathcal{S}' = ([i,j]^m_{\text{mul}:k-1}.\delta, \mathcal{P}, E_{L_{\mathcal{S}}})$ such that $E(\mathcal{S}') = M_{i,j,k-2}^\text{m:2} + E_{L_{\mathcal{S}}} + \sum \limits_{[m,n]^t \in \delta} E([m,n]^t)$. 
		
		2) If the base $k-1$ is paired with another base $d \in [i,k-2]$,  this will result in the new partial structure: 
		
		$\mathcal{S}' = ([i,d-1]^m.[d,k-1]^b.\delta, \mathcal{P},  \Delta G_\text{bp}^\text{multi} + (j-k + 1) \Delta G_\text{nt}^\text{multi} + E_{L_{\mathcal{S}}})$ such that $E(\mathcal{S}') =   M^m_{i,d-1} +  M^b_{d,k-1} + \Delta G_\text{bp}^\text{multi} + (j-k + 1)\Delta G_\text{nt}^\text{multi} +E_{L_{\mathcal{S}}} 
		+ \sum \limits_{[m,n]^t \in \delta} E([m,n]^t)$.

	\end{itemize}
	
	Then the acceptance criteria will be checked after each sub-case. Now, with the aid of the introduced new auxiliary matrices $M^\text{mul:2}$, the backtracking algorithm   checks up to $\mathcal{O}(N)$ refined structures, hence saves up to  $\mathcal{O}(N)$ refined structures to $\mathcal{R}_u$, in the worst case.
	
\end{enumerate}

\begin{remark}
	For any $\mathcal{S}$ and  $\mathcal{S}'$ such that $\mathcal{S}'$ is a refinement of $\mathcal{S}$ based on refinement rules described above, $E(\mathcal{S}) \leq E(\mathcal{S}')$. Also, note that the form of the new refined structure $\mathcal{S}'$ in each case of refinement cases is different, and hence leads to a different fully specified structure which guarantees that each secondary structure $S_u$ encountered during the backtracking is scanned exactly once. 
\end{remark} 

\begin{remark}\label{remark:spaceE}
	For all the cases above where  $\mathcal{S}'$ is refinement of $\mathcal{S}$,  the stack $\delta$ and the base pairs set $\mathcal{P}$ are parts (common) of each refined structure  $\mathcal{S}'$,  hence it is enough to save them once, which  takes $\mathcal{O}(N)$ space, and for each refined structure $\mathcal{S}'$, we need to save only the additional base pairs (one base pair in the worst case, hence $\mathcal{O}(1)$ space), or the new segments (two segments in the worst case, hence $\mathcal{O}(1)$ space) that are pushed on the top of $\delta$ based on the refinement case. In total saving the all $\mathcal{O}(N)$ refined structures that are generated from all cases requires $\mathcal{O}(N)$ space.

\end{remark}  

\begin{remark}\label{remark:timeE}
	For all the cases above where  $\mathcal{S}'$ is refinement of $\mathcal{S}$,  $\mathcal{H} = E_{L_{\mathcal{S}}} + \sum_{[m,n]^t \in \delta} E([m,n]^t)$ 
	is repeatedly used to compute $E(\mathcal{S}')$ again and again, hence it is enough to compute it once, which takes only $\mathcal{O}(N)$ time. 
	In total, if $\mathcal{H}$ is pre-computed once in this way, computing $E(\mathcal{S}')$ takes $\mathcal{O}(1)$ time. 
\end{remark}

After scanning the secondary structure $S_u$ completely, the partially specified structure $\mathcal{S}' \in \mathcal{R}_u$, such that $E(\mathcal{S}') =  \min\limits_{\mathcal{S} \in \mathcal{R}_u}\{ E(\mathcal{S})\}$ will be computed and saved. And a new partially specified structure $\mathcal{S}$ is chosen as follows:

\begin{equation} \label{eq:newS}
	\mathcal{S} = \min\limits_{z \in \left\{1,\ldots,u\right\}}\left\{\min\limits_{\mathcal{S}' \in \mathcal{R}_z}\left\{ E(\mathcal{S}')\right\}\right\} 
\end{equation}

Where $\{S_1,\ldots,S_u\}$ is the set of distinct secondary structures that are completely scanned until the moment, where $u<\mathcal{U}$. Then the minimum ($\min_{\mathcal{S}' \in \mathcal{R}_z}\left\{ E(\mathcal{S}')\right\}$) of the array $\mathcal{R}_{z'}$ where $\mathcal{S}$ is choose from, will be computed again and saved for future iterations. 

\begin{remark}\label{remark:newS}
	\cref{eq:newS} guarantees sequential scanning of the backtracking algorithm through energy levels without skipping any potential structure, due to free energy minimality. Also, note that for all $z \in \{1,\ldots,u\}$, the  $\min_{\mathcal{S}' \in \mathcal{R}_z}\left\{ E(\mathcal{S}')\right\}$ is already computed and saved, as in each iteration we choose only the minimum over the set of all minimum energies of each array, so we lose only one the minimum of some array $\mathcal{R}_{z'}$, so $\mathcal{R}_{z'}$ is the only one we need to compute its minimum again in $\mathcal{O}(N^2)$ time.        
\end{remark} 

The same refinement process starts again with that new selected partially specified structure $\mathcal{S}$. This backtracking procedure continues in this way until one of the three following conditions occurs first: 
\begin{enumerate}
	\item  The algorithm scans an {\em asymmetric} secondary structure $S_u$, then the true MFE $= \Delta G (S_u)$, as a direct consequence of $\Delta G (S_u) \leq \mathcal{B}$, where we recall that $\mathcal{B}$ was the best candidate value for true MFE, or 
	\item  The algorithm exceeds the upper bound $\mathcal{U}$ of the number of symmetric secondary structures to be scanned, then the true MFE $= \mathcal{E}$, the energy of the last scanned energy level which is also the \snMFE of that last completely scanned symmetric secondary structure $S_u$, as a direct consequence of \cref{lem:polyub,lem:sand,lem:sand2} (meaning we have two symmetric structures of  \snMFE equal to $\mathcal{E}$ with the same admissible cut, hence we can make a new pizza: an asymmetric secondary structure of true MFE $\mathcal{E}$), or
	\item  The algorithm starts scanning a new energy level $\mathcal{E}' > \mathcal{B}$, then the true MFE $= \mathcal{B}$, as $\mathcal{B}$ is the best candidate for the true MFE that we got from a previously scanned symmetric secondary structure.   
\end{enumerate}
Whichever of the three cases occurs, the true MFE is returned (and a secondary structure with that true MFE is constructed).

\subsection{Time and space complexity analysis of the backtracking algorithm}\label{sec:backtime}

The backtracking algorithm needs to scan up to $\mathcal{U} = \mathcal{O}(N^2)$ secondary structures in the worst case ({\cref{lem:polyub}}), so the total time complexity of the backtracking algorithms is $\mathcal{O}(\mathcal{U} \mathcal{W})$, where $\mathcal{W}$ is the time complexity of scanning only one secondary structure and setting up the scene for the next iteration by choosing the new partially specified structure required to scan the next secondary structure. 

\paragraph{Analysis for scanning only one secondary structure in the backtracking algorithm.}
To scan (construct) one secondary structure $S_u$, we need in the worst case  $N = \mathcal{O}(N)$ refinement steps, as each step either ignores a base or forms a base pair. We showed in our analysis, in \cref{sec:backhigh}, and based on $t \in \{\square,b,m\}$, that each step checks up to $\mathcal{O}(N)$ refined structures and saves up to $\mathcal{O}(N)$ refined structures to $\mathcal{R}_u$, the array of refined structures associated with $S_u$. In total, by \cref{remark:timeE}, scanning one secondary structure takes $\mathcal{O}(N^2)$, as $\mathcal{R}_u$ contains $\mathcal{O}(N^2)$ structure, hence computing the minimum of $\mathcal{R}_u$ takes $\mathcal{O}(N^2)$ time. 

The last step is to see what is the time complexity of choosing the new partially specified structure $\mathcal{S}$ for the next iteration based on \cref{eq:newS}, from \cref{remark:newS} this step takes $\mathcal{O}(N^2)$ time, as the minimum of all the $\mathcal{U} = \mathcal{O}(N^2)$ arrays is already computed and stored. 

So, in total scanning one secondary structure and setting up the scene for the next iteration by choosing the new partially specified structure $\mathcal{S}$ takes $\mathcal{O}(N^2)$ time.

\cref{remark:spaceE} shows that each array $\mathcal{R}_u$ requires $\mathcal{O}(N^2)$ space, hence in total the backtracking algorithm requires $\mathcal{O}(N^4)$ space to store the all $\mathcal{R}_u$ such that $u \in \{1, \ldots, \mathcal{U}\}$. This analysis leads to the following result.

\begin{lemma}\label{lem:BTtime}
	The running time  of the backtracking algorithm,   \cref{algo:2} and \cref{sec:backhigh},
	for a set of $c = \mathcal{O}(1)$ DNA or RNA strands of total length  $N$ bases, is $\mathcal{O}(N^4(c-1)!)$, and it requires $\mathcal{O}(N^4)$ space.  
\end{lemma}

\begin{remark}\label{remark:newT}
	We should note that, changing the strategy, used in \cref{eq:newS}, for choosing the new partially specified structure $\mathcal{S}$ can help in reducing the space complexity of the backtracking algorithm from $\mathcal{O}(N^4)$ to $\mathcal{O}(N^3)$ (the same space complexity as \snMFE \cref{algo:1}) with the trade off increasing the time complexity to be $\mathcal{O}(N^4 \log N(c-1)!)$ instead of $\mathcal{O}(N^4 (c-1)!)$. As we know that we need to scan only $\mathcal{U} = \mathcal{O}(N^2)$ secondary structures, so we do not need to store all elements of arrays $\mathcal{R}_z$ where $z \in \{1,\ldots,u\}$. 
	Only the minimum $\mathcal{U}$ candidates should be stored. By sorting $\mathcal{R}_u$, the array of partially specified structures that we obtain after constructing the secondary structure $S_u$ (the secondary structure number $u$, of the worst case $\mathcal{U}$). $\mathcal{R}_u$ can be sorted in $\mathcal{O}(N^2 \log N)$ time then merged in $\mathcal{O}(N^2)$ time with the already sorted array, that we obtain cumulatively through time from the previous iterations.        
\end{remark} 

As we noted before in \cref{lem:even}, that the (bad) quadratic upper bound $\mathcal{U}$ in \cref{lem:ub2} is very restricted and rare, and 
for systems of $c$ strands ($k$ strand types) where the repetition number of some strand type is odd, \cref{lem:even} shows that the upper bound $\mathcal{U}$ is linear. Hence, the time complexity of the backtracking algorithm is $\mathcal{O}(N^3(c-1)!)$ and it requires $\mathcal{O}(N^3)$ space.

\section{Time and space analysis of MFE algorithm}\label{sec:analysis}

\main*

\begin{proof}
	Dirks et al's \snMFE algorithm runs in time $\mathcal{O}(N^4 (c-1)!)$ and space $\mathcal{O}(N^2)$~\cite{dirks2007thermodynamic}. 
	In \cref{algo:1}, we give their  \snMFE pseudocode but augmented with three  matrices $M^\text{b:int}$, $M^\text{b:mul}$, and $M^\text{m:2}$,  
	with no asymptotic change to run time but an increase to $\mathcal{O}(N^3)$ space. 
	Also, by \cref{lem:BTtime} the time complexity of our backtracking algorithm, \cref{algo:2}, is $\mathcal{O}(N^4 (c-1)!)$, and the space complexity is $\mathcal{O}(N^4)$.
	
	Hence after running both algorithms we get an 
	$\mathcal{O}(N^4 (c-1)!)$ algorithm for the 
	MFE unpseudoknotted secondary structure prediction problem, including rotational symmetry, with  $\mathcal{O}(N^4)$ as space complexity.
	
	By \cref{remark:newT} we can get another alternative algorithm of  $\mathcal{O}(N^4 \log N(c-1)!)$ time and $\mathcal{O}(N^3)$ space. 
\end{proof}

	\subsection*{Acknowledgements} 
	We thank Constantine Evans for his helpful comments on the statistical mechanics origin of the MFE rotational symmetry penalty, and Dave Doty for comments on the manuscript. Ahmed Shalaby would like to thank Dvořák for composing his masterpiece, Symphony No.~9 ``From the New World''. 
	
	\bibliographystyle{plainurl}   
	\bibliography{bib}

\begin{thebibliography}{10}

\bibitem{akutsu2000dynamic}
Tatsuya Akutsu.
\newblock Dynamic programming algorithms for {RNA} secondary structure
  prediction with pseudoknots.
\newblock {\em Discrete Applied Mathematics}, 104(1-3):45--62, 2000.

\bibitem{atkins2023atkins}
Peter~William Atkins, Julio De~Paula, and James Keeler.
\newblock {\em Atkins' physical chemistry}.
\newblock Oxford university press, 2023.

\bibitem{boehmer2024rna}
Kimon Boehmer, Sarah~J Berkemer, Sebastian Will, and Yann Ponty.
\newblock Rna triplet repeats: Improved algorithms for structure prediction and
  interactions.
\newblock 2024.
\newblock 24th Workshop on Algorithms in Bioinformatics (WABI), to appear.
\newblock URL: \url{https://hal.science/hal-04589903}.

\bibitem{bormashenko2019entropy}
Edward Bormashenko.
\newblock Entropy, information, and symmetry: Ordered is symmetrical.
\newblock {\em Entropy}, 22(1):11, 2019.

\bibitem{brualdi1977introductory}
Richard~A Brualdi.
\newblock {\em Introductory combinatorics}.
\newblock Pearson Education India, 1977.

\bibitem{Chatterjee2017}
Gourab Chatterjee, Neil Dalchau, Richard~A. Muscat, Andrew Phillips, and Georg
  Seelig.
\newblock A spatially localized architecture for fast and modular {DNA}
  computing.
\newblock {\em Nature Nanotechnology}, 12(9):920--927, Sep 2017.
\newblock \href {https://doi.org/10.1038/nnano.2017.127}
  {\path{doi:10.1038/nnano.2017.127}}.

\bibitem{chen2009n}
Ho-Lin Chen, Anne Condon, and Hosna Jabbari.
\newblock An {$O(n^5)$} algorithm for {MFE} prediction of kissing hairpins and
  4-chains in nucleic acids.
\newblock {\em Journal of Computational Biology}, 16(6):803--815, 2009.

\bibitem{churkin2018design}
Alexander Churkin, Matan~Drory Retwitzer, Vladimir Reinharz, Yann Ponty,
  J{\'e}r{\^o}me Waldisp{\"u}hl, and Danny Barash.
\newblock Design of {RNAs}: comparing programs for inverse {RNA} folding.
\newblock {\em Briefings in bioinformatics}, 19(2):350--358, 2018.

\bibitem{condon2021predicting}
Anne Condon, Monir Hajiaghayi, and Chris Thachuk.
\newblock Predicting minimum free energy structures of multi-stranded nucleic
  acid complexes is {APX}-hard.
\newblock In {\em 27th International Conference on DNA Computing and Molecular
  Programming (DNA 27)}. Schloss Dagstuhl-Leibniz-Zentrum f{\"u}r Informatik,
  2021.

\bibitem{dirks2007thermodynamic}
Robert~M Dirks, Justin~S Bois, Joseph~M Schaeffer, Erik Winfree, and Niles~A
  Pierce.
\newblock Thermodynamic analysis of interacting nucleic acid strands.
\newblock {\em SIAM review}, 49(1):65--88, 2007.

\bibitem{dirks2003partition}
Robert~M Dirks and Niles~A Pierce.
\newblock A partition function algorithm for nucleic acid secondary structure
  including pseudoknots.
\newblock {\em Journal of computational chemistry}, 24(13):1664--1677, 2003.

\bibitem{dirks2004algorithm}
Robert~M Dirks and Niles~A Pierce.
\newblock An algorithm for computing nucleic acid base-pairing probabilities
  including pseudoknots.
\newblock {\em Journal of computational chemistry}, 25(10):1295--1304, 2004.

\bibitem{nuad}
David Doty and Benjamin Lee.
\newblock nuad: Nucleic acid designer, March 2022.
\newblock Uses, and generalises, sequence design principles from
  \cite{algoSST}.
\newblock URL: \url{https://github.com/UC-Davis-molecular-computing/nuad}.

\bibitem{celltimer2017fernshulman}
Joshua Fern and Rebecca Schulman.
\newblock Design and characterization of dna strand-displacement circuits in
  serum-supplemented cell medium.
\newblock {\em ACS Synthetic Biology}, 6(9):1774--1783, 2017.
\newblock \href {https://doi.org/10.1021/acssynbio.7b00105}
  {\path{doi:10.1021/acssynbio.7b00105}}.

\bibitem{fornace2020unified}
Mark~E Fornace, Nicholas~J Porubsky, and Niles~A Pierce.
\newblock A unified dynamic programming framework for the analysis of
  interacting nucleic acid strands: enhanced models, scalability, and speed.
\newblock {\em ACS Synthetic Biology}, 9(10):2665--2678, 2020.

\bibitem{geary2014single}
Cody Geary, Paul~WK Rothemund, and Ebbe~S Andersen.
\newblock A single-stranded architecture for cotranscriptional folding of {RNA}
  nanostructures.
\newblock {\em Science}, 345(6198):799--804, 2014.

\bibitem{hofacker2012symmetric}
Ivo~L Hofacker, Christian~M Reidys, and Peter~F Stadler.
\newblock Symmetric circular matchings and {RNA} folding.
\newblock {\em Discrete mathematics}, 312(1):100--112, 2012.

\bibitem{jabbari2018knotty}
Hosna Jabbari, Ian Wark, Carlo Montemagno, and Sebastian Will.
\newblock Knotty: efficient and accurate prediction of complex {RNA} pseudoknot
  structures.
\newblock {\em Bioinformatics}, 34(22):3849--3856, 2018.

\bibitem{viennaRNA}
Ronny Lorenz, Stephan~H Bernhart, Christian H{\"o}ner~zu Siederdissen, Hakim
  Tafer, Christoph Flamm, Peter~F Stadler, and Ivo~L Hofacker.
\newblock {ViennaRNA} package 2.0.
\newblock {\em Algorithms for molecular biology}, 6:1--14, 2011.

\bibitem{lyngso2000pseudoknots}
Rune~B Lyngs{\o} and Christian~NS Pedersen.
\newblock Pseudoknots in {RNA} secondary structures.
\newblock In {\em Proceedings of the fourth annual international conference on
  Computational molecular biology}, pages 201--209, 2000.

\bibitem{lyngso2000rna}
Rune~B Lyngs{\o} and Christian~NS Pedersen.
\newblock {RNA} pseudoknot prediction in energy-based models.
\newblock {\em Journal of computational biology}, 7(3-4):409--427, 2000.

\bibitem{lyngso1999fast}
Rune~B Lyngs{\o}, Michael Zuker, and CN~Pedersen.
\newblock Fast evaluation of internal loops in {RNA} secondary structure
  prediction.
\newblock {\em Bioinformatics (Oxford, England)}, 15(6):440--445, 1999.

\bibitem{mathews1999expanded}
David~H Mathews, Jeffrey Sabina, Michael Zuker, and Douglas~H Turner.
\newblock Expanded sequence dependence of thermodynamic parameters improves
  prediction of {RNA} secondary structure.
\newblock {\em Journal of molecular biology}, 288(5):911--940, 1999.

\bibitem{mccaskill1990equilibrium}
John~S McCaskill.
\newblock The equilibrium partition function and base pair binding
  probabilities for {RNA} secondary structure.
\newblock {\em Biopolymers: Original Research on Biomolecules},
  29(6-7):1105--1119, 1990.

\bibitem{nicholson2012introduction}
W~Keith Nicholson.
\newblock {\em Introduction to abstract algebra}.
\newblock John Wiley \& Sons, 2012.

\bibitem{nussinov1980fast}
Ruth Nussinov and Ann~B Jacobson.
\newblock Fast algorithm for predicting the secondary structure of
  single-stranded {RNA}.
\newblock {\em Proceedings of the National Academy of Sciences},
  77(11):6309--6313, 1980.

\bibitem{nussinov1978algorithms}
Ruth Nussinov, George Pieczenik, Jerrold~R Griggs, and Daniel~J Kleitman.
\newblock Algorithms for loop matchings.
\newblock {\em SIAM Journal on Applied mathematics}, 35(1):68--82, 1978.

\bibitem{squareRoot}
Lulu Qian and Erik Winfree.
\newblock Scaling up digital circuit computation with {DNA} strand displacement
  cascades.
\newblock {\em Science}, 332(6034):1196--1201, 2011.

\bibitem{reeder2004design}
Jens Reeder and Robert Giegerich.
\newblock Design, implementation and evaluation of a practical pseudoknot
  folding algorithm based on thermodynamics.
\newblock {\em BMC bioinformatics}, 5:1--12, 2004.

\bibitem{rivas1999dynamic}
Elena Rivas and Sean~R Eddy.
\newblock A dynamic programming algorithm for {RNA} structure prediction
  including pseudoknots.
\newblock {\em Journal of molecular biology}, 285(5):2053--2068, 1999.

\bibitem{santalucia1998unified}
John SantaLucia~Jr.
\newblock A unified view of polymer, dumbbell, and oligonucleotide {DNA}
  nearest-neighbor thermodynamics.
\newblock {\em Proceedings of the National Academy of Sciences},
  95(4):1460--1465, 1998.

\bibitem{santa}
John SantaLucia~Jr and Donald Hicks.
\newblock The thermodynamics of {DNA} structural motifs.
\newblock {\em Annu. Rev. Biophys. Biomol. Struct.}, 33:415--440, 2004.

\bibitem{sawada2003fast}
Joe Sawada.
\newblock A fast algorithm to generate necklaces with fixed content.
\newblock {\em Theoretical Computer Science}, 301(1-3):477--489, 2003.

\bibitem{seelig2006enzyme}
Georg Seelig, David Soloveichik, David~Yu Zhang, and Erik Winfree.
\newblock Enzyme-free nucleic acid logic circuits.
\newblock {\em science}, 314(5805):1585--1588, 2006.

\bibitem{silbey2022physical}
Robert~J Silbey, Robert~A Alberty, George~A Papadantonakis, and Moungi~G
  Bawendi.
\newblock {\em Physical chemistry}.
\newblock John Wiley \& Sons, 2022.

\bibitem{cargoSorting}
Anupama~J. Thubagere, Wei Li, Robert~F. Johnson, Zibo Chen, Shayan Doroudi,
  Yae~Lim Lee, Gregory Izatt, Sarah Wittman, Niranjan Srinivas, Damien Woods,
  Erik Winfree, and Lulu Qian.
\newblock A cargo-sorting {DNA} robot.
\newblock {\em Science}, 357(6356), 2017.

\bibitem{tinoco}
Ignacio Tinoco, Olke~C Uhlenbeck, and Mark~D Levine.
\newblock Estimation of secondary structure in ribonucleic acids.
\newblock {\em Nature}, 230(5293):362--367, 1971.

\bibitem{uemura1999tree}
Yasuo Uemura, Aki Hasegawa, Satoshi Kobayashi, and Takashi Yokomori.
\newblock Tree adjoining grammars for {RNA} structure prediction.
\newblock {\em Theoretical computer science}, 210(2):277--303, 1999.

\bibitem{SIMDDNA}
Boya Wang, Siyuan~Stella Wang, Cameron Chalk, Andrew~D Ellington, and David
  Soloveichik.
\newblock Parallel molecular computation on digital data stored in {DNA}.
\newblock {\em Proceedings of the National Academy of Sciences},
  120(37):e2217330120, 2023.

\bibitem{waterman1985dynamic}
Michael~S Waterman and Thomas~H Byers.
\newblock A dynamic programming algorithm to find all solutions in a
  neighborhood of the optimum.
\newblock {\em Mathematical Biosciences}, 77(1-2):179--188, 1985.

\bibitem{waterman1986rapid}
Michael~S Waterman and Temple~F Smith.
\newblock Rapid dynamic programming algorithms for {RNA} secondary structure.
\newblock {\em Advances in Applied Mathematics}, 7(4):455--464, 1986.

\bibitem{west2001introduction}
Douglas~Brent West et~al.
\newblock {\em Introduction to graph theory}, volume~2.
\newblock Prentice hall Upper Saddle River, 2001.

\bibitem{woo2011programmable}
Sungwook Woo and Paul~WK Rothemund.
\newblock Programmable molecular recognition based on the geometry of {DNA}
  nanostructures.
\newblock {\em Nature chemistry}, 3(8):620, 2011.

\bibitem{algoSST}
Damien Woods, David Doty, Cameron Myhrvold, Joy Hui, Felix Zhou, Peng Yin, and
  Erik Winfree.
\newblock Diverse and robust molecular algorithms using reprogrammable {DNA}
  self-assembly.
\newblock {\em Nature}, 567(7748):366--372, 2019.

\bibitem{wuchty1999complete}
Stefan Wuchty, Walter Fontana, Ivo~L Hofacker, and Peter Schuster.
\newblock Complete suboptimal folding of {RNA} and the stability of secondary
  structures.
\newblock {\em Biopolymers: Original Research on Biomolecules}, 49(2):145--165,
  1999.

\bibitem{zhang2011dynamic}
David~Yu Zhang and Georg Seelig.
\newblock Dynamic {DNA} nanotechnology using strand-displacement reactions.
\newblock {\em Nature chemistry}, 3(2):103--113, 2011.

\bibitem{mfold}
Michael Zuker.
\newblock Mfold web server for nucleic acid folding and hybridization
  prediction.
\newblock {\em Nucleic acids research}, 31(13):3406--3415, 2003.

\bibitem{zukerrna}
Michael Zuker and David Sankoff.
\newblock {RNA} secondary structures and their prediction.
\newblock {\em Bulletin of mathematical biology}, 46:591--621, 1984.

\bibitem{zukeroptimal}
Michael Zuker and Patrick Stiegler.
\newblock Optimal computer folding of large {RNA} sequences using
  thermodynamics and auxiliary information.
\newblock {\em Nucleic acids research}, 9(1):133--148, 1981.

\end{thebibliography}
	
	\newpage
	\appendix
\section{Appendix: Useful lemmas}\label{sec:lemmasApp}
\begin{lemma}[factors of $\pi$]\label{lem:factors}
	For any nonempty circular permutation $\pi$ 
	and any prefix $y$ of $\pi$ that is not its fundamental component $x$, 
	such that $|y|>|x|$, 
	if $\pi = y^n$ then $|x|$ divides $|y|$.    
\end{lemma}
\begin{proof}
	Suppose for contradiction that such $y$ and $n$ exist, so $x^{v(\pi)} = y^n$ which means that $v(\pi)|x| = n|y|$. As $|x|$ does not divide $|y|$, from division algorithm we can write $|y| = a |x| + b$ for some remainder $0<b<|x|$.
	
	Suppose the infinite string $W$ that is just an infinite repeat of $\pi$ and hence an infinite repeat of $x$ and $y$. Let's suppose a simple function $f(i)$ that return the character at index $i$ of $W$, as $\pi = x^{v(\pi)} = y^n$ that means that $f(i) = f(i+\alpha|x|) = f(i+\beta|y|)$ for any $\alpha$ and $\beta$. Then $f(i) = f(i+ a\beta |x| + \beta b)$. For $\beta = 1$ that implies that $f(i) = f(i+b)$, which implies the  existence of a smaller repeating prefix of $\pi$ since $b<|x|$ which contradicts the fact that $x$ is the fundamental component of $\pi$.   
\end{proof}

The following lemma restricts us to deal with only specific and constant number of different folding rotational symmetries, and hence constant different symmetry corrections in total. 
\begin{lemma}\label{lem:div}
	If $S$ is $R$-fold rotational symmetric secondary structure, with a specific circular permutation $\pi$, then $R$ must be a divisor of $v(\pi)$.
\end{lemma}

\begin{proof}
	As $R = |H|$ a subgroup of $G^\pi$, and from Lagrange theorem for finite groups~\cite{nicholson2012introduction}, which says that for any finite group $G$, the order of every subgroup of $G$ divides the order of $G$. So, R divides $|G^\pi| = v(\pi)$. Also, from the fundamental theorem of cyclic groups, as we know that $H$ is a subgroup of the cyclic group $G^\pi$, then $H = \langle \rho^d \rangle$ for some $\rho^d \in G^\pi$, $H$ is generated by $\rho^d$. Hence, $H$ is the \emph{unique} subgroup of order $|H|= v(\pi)/d$. 
\end{proof}

\begin{lemma} \label{lem:nobase}
	For any connected unpseudoknotted secondary structure $S$, if there exists at least one base pair $(i,j)$ such that $\llbracket i,j \rrbracket > \frac{N}{R}$, then $S$ can not be  $R$-fold rotationally symmetric. 
\end{lemma}
\begin{proof}
	For all $R>2$, suppose for the sake of contradiction that $S$ is $R$-fold rotational symmetric secondary structure and such base pair $(i,j)$ exists, from symmetry this implies the existence of another $(R-1)$ different base pairs each with same segment length as $\llbracket i,j \rrbracket$, so the total length of the system must be higher than number of bases $N$, so at least two base pairs must intersect forming a pseudoknot, which contradicts the fact that $S$ is pseudoknot-free. 
	
	For $R=2$, suppose such base pair $(i,j)$ exists, then this can only happen if $(i,j)$ is a central base pair, a diameter in $\PolySpi$, such that $\llbracket i,j \rrbracket = \frac{N}{2}+1 >  \frac{N}{2}$. This is impossible as the symmetry implies that $i$ and $j$ are of the same type, that there exists a base which is complement to itself.  
\end{proof}

 \newpage

\section{Appendix: \SymnMFE  (\snMFE) algorithm}\label{sec:AlgoMFE}

\cref{algo:1}, shown in \cref{fig:mfe}, computes \snMFE  for a constant number, $c=\mathcal{O}(1)$, of interacting nucleic acid strands. 
We should note that: \cref{algo:1} is a straightforward conversion of the partition function algorithm from Dirks et al. \cite{dirks2007thermodynamic}. 
\cref{algo:1} ignores rotational symmetry, 
if the predicted \snMFE structure from this algorithm happens to be asymmetric, 
then the output of \cref{algo:1} is the true MFE as there is no symmetry correction penalty for asymmetric secondary structures.
However, if the \snMFE structure is an $R$-fold symmetric secondary structure, then its free energy must be corrected by by $+k_\mathrm{B} T \log R$, a positive value, then it is not guaranteed that the \snMFE will be the true MFE without scanning all secondary structures in the window of $k_\mathrm{B} T \log R$ above \snMFE, and applying any needed symmetry corrections to the free energy of each secondary structure that lies in that window. Wuchty et al. \cite{wuchty1999complete} showed that this window of secondary structures could scale exponentially with $N$, which shows why this strategy fails. But in this work, we proved that only a polynomial number of these structures are enough for predicting the true MFE. 

We introduced new three-dimensional matrices $M^\text{b:int}, M^\text{b:mul}$,  $M^\text{m:2}$ to help in reducing the time complexity of the backtracking algorithm. For any segment $[i,j]^b$, if $i+2 \leq k \leq j-1$, $M_{i,j,k}^\text{b:int}$ will contain the minimum free energy attainable from the segment $[i,j]^b$ such that all bases $d$, such that $k<d<j$ are unpaired, and there exist exactly one base pair $(m,n)$, such that $i+1 \leq m < n \leq k$, $(i,j)$ and $(m,n)$ are forming together an internal loop. The same for $M_{i,j,k}^\text{b:mul}$, except that there exist more than one base pair $(m,n)$, such that $i+1 \leq m < n \leq k$, forming together a multiloop with $(i,j)$.

For any segment $[i,j]^m$, if $i+1 \leq k \leq j$, $M_{i,j,k}^\text{m:2}$ contains the minimum free energy attainable from the segment $[i,j]^m$ such that all bases $d$, such that $k<d \leq j$ are unpaired, and there exist more than one base pair $(m,n)$, such that $i \leq m < n \leq k$, forming together a multiloop with $(i,j)$.

\begin{figure}[H]
	\centering\includegraphics[width=1\textwidth]{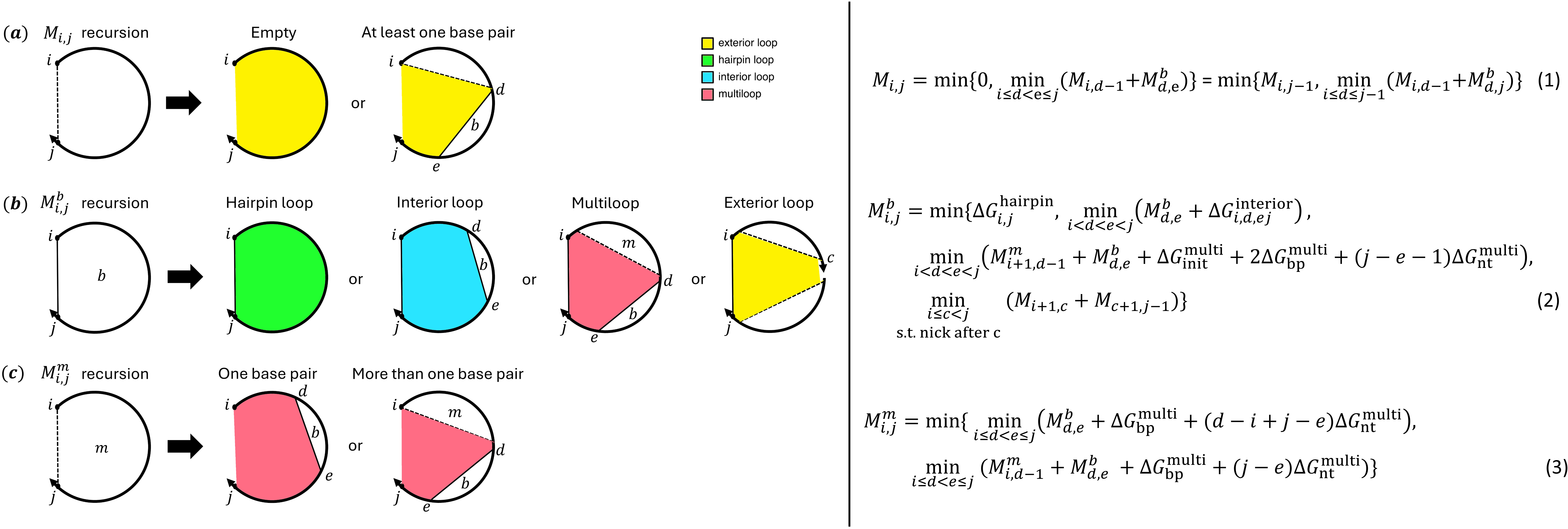}
	\caption{\snMFE dynamic program recursion diagrams (left) and recursion equations (right). A solid straight line indicates a base pair and
		a dashed line demarcates a region without implying that the connected bases are paired. Shaded regions correspond to loop free energies that are	explicitly incorporated at the current level of recursion. See \cite{dirks2003partition, fornace2020unified} for  full details.  
	}\label{fig:mfe}
\end{figure}

\begin{algorithm}[H] 
	\caption{\small\SymnMFE (\snMFE) algorithm pseudocode  that takes as input: $c=\mathcal{O}(1)$ strands with total number of bases (length) $N$ and strand ordering $\pi$. 
		Runs in  $\mathcal{O}(N^4)$ time and $\mathcal{O}(N^3)$ space with
		recursive calls illustrated in \cref{fig:mfe}. Nicks between strands are denoted by half indices (e.g.~$x+ \frac{1}{2}$). 
		The function $\eta[i+ \frac{1}{2}, j+\frac{1}{2}]$ returns the number of nicks in the interval $[i+ \frac{1}{2}, j+\frac{1}{2}]$. 
		The shorthand $\eta[i+ \frac{1}{2}]$ is equivalent to $\eta[i+ \frac{1}{2}, i+\frac{1}{2}]$ and by convention, $\eta[i+ \frac{1}{2}, i-\frac{1}{2}] =0$.
	} \label{algo:1}
	\begin{algorithmic}[1]
		\footnotesize
		\State Initialize $M, M^b, M^m, M^\text{b:int}, M^\text{b:mul}$,  $M^\text{m:2}$  by setting all values to $+\infty$, except $M_{i,i-1} = 0$ for all $i=1,\ldots,N$
		
		\For{$l \gets 1 \ldots N$}
		\For{$i \gets 1 \ldots N-l+1$}
		\State $j = i+l-1$
		
		\Comment{$M^b$ recursion equations} 
		\If{$\eta[i+\frac{1}{2}, j-\frac{1}{2}] ==0$}
		$M_{i,j}^b =\Delta G_{i,j}^\text{hairpin}$
		\Comment{hairpin loop requires no nicks}

		\EndIf 
		
		\State $\text{min}_\text{int}^b =  \text{min}_\text{mul}^b = \text{min}_\text{mul}^m = +\infty$

		\For{$e \gets i+2 \ldots j-1$}
		\Comment{loop over all possible $3'$-most pairs $(d,e)$} 
		
		\For{$d \gets i+1 \ldots e-1$} 
		
		\If{$\eta[i+\frac{1}{2}, d-\frac{1}{2}] ==0$ and $\eta[e+\frac{1}{2}, j-\frac{1}{2}] ==0$}
		$M_{i,j}^b = \min \{ M_{i,j}^b, M_{d,e}^b + \Delta G_{i,d,e,j}^\text{interior} \}$
		
		\If{$(M_{d,e}^b + \Delta G_{i,d,e,j}^\text{interior}) < \text{min}_\text{int}^b $}
		$\text{min}_\text{int}^b = M_{d,e}^b + \Delta G_{i,d,e,j}^\text{interior} $
		
		\EndIf

		\EndIf 
		
		\If{$\eta[e+\frac{1}{2}, j-\frac{1}{2}] ==0$ and $\eta[i+\frac{1}{2}] ==0$ and $\eta[d-\frac{1}{2}] ==0$}
		\Comment{multiloop: no nicks} 
		
		\State  $M_{i,j}^b = \min \{ M_{i,j}^b, M^b_{d,e} + M^m_{i+1,d-1} + \Delta G_\text{init}^\text{multi} + 2\Delta G_\text{bp}^\text{multi} + (j-e-1)\Delta G_\text{nt}^\text{multi}\}$
		
		\If{$(  M^m_{i+1,d-1} +  M^b_{d,e} + \Delta G_\text{init}^\text{multi} + 2\Delta G_\text{bp}^\text{multi} + (j-e-1)\Delta G_\text{nt}^\text{multi}) < \text{min}_\text{mul}^b$}
		
		\State $\text{min}_\text{mul}^b =   M^m_{i+1,d-1} + M^b_{d,e} + \Delta G_\text{init}^\text{multi} + 2\Delta G_\text{bp}^\text{multi} + (j-e-1)\Delta G_\text{nt}^\text{multi} $
		
		\EndIf
		
		\EndIf
		\EndFor
		
		\State $M_{i,j,e}^\text{b:int} = \text{min}_\text{int}^b$; \ \  
		$M_{i,j,e}^\text{b:mul} = \text{min}_\text{mul}^b$ \Comment{for the new auxiliary matrices}

		\EndFor

		\For{$x \in \{i, \ldots,j+1\}$ s.t. $\eta[x+\frac{1}{2}] = 1$} 
		\Comment{loop over all nicks $\in [i+\frac{1}{2}, j-\frac{1}{2}]$} 
		
		\If{($\eta[i+\frac{1}{2}] == 0$ and $\eta[j-\frac{1}{2}] == 0$) or ($i==j-1$) or  ($x==i$ and $\eta[j-\frac{1}{2}] == 0$) or
			\par \hskip\algorithmicindent  
			\hspace{6 mm} ($x==j-1$ and $\eta[i+\frac{1}{2}] == 0$)}
		
		\State  $M_{i,j}^b = \min \{ M_{i,j}^b, M_{i+1,x} + M_{x+1,j-1} \}$\Comment{exterior loops} 
		
		\EndIf
		\EndFor

		\\	
		\Comment{$M, M^m$ recursion equations} 
		
		\If{$\eta[i+\frac{1}{2}, j-\frac{1}{2}] == 0$} 
		$M_{i,j} = 0$\Comment{empty substructure} 
		
		\EndIf
		
		\For{$e \gets i+1 \ldots j$} \Comment{loop over all possible $3'$-most pairs $(d,e)$}
		
		\For{$d \gets i \ldots e-1$}
		
		\If{$\eta[e+\frac{1}{2}, j-\frac{1}{2}] == 0$} 
		
		\If{$\eta[d-\frac{1}{2}] == 0$ or $d==i$}
		$M_{i,j} = \min \{M_{i,j}, M_{i,d-1} + M_{d,e} \}$
		\EndIf
		
		\If{$\eta[i+\frac{1}{2}, d-\frac{1}{2}] == 0$}
		
		\State $M_{i,j}^m = \min \{M_{i,j}^m, M^b_{d,e} + \Delta G_\text{bp}^\text{multi} + (d-i + j-e)\Delta G_\text{nt}^\text{multi}\}$
		\Comment{single base pair}
		
		\EndIf
		\If{$\eta[d-\frac{1}{2}] == 0$}
		\State $M_{i,j}^m = \min \{M_{i,j}^m, M^b_{d,e} + M^m_{i,d-1} + \Delta G_\text{bp}^\text{multi} + (j-e)\Delta G_\text{nt}^\text{multi}\}$
		\Comment{more than one base pair}
		
		\If{$( M^m_{i,d-1} +  M^b_{d,e} +  \Delta G_\text{bp}^\text{multi} + (j-e)\Delta G_\text{nt}^\text{multi}) < \text{min}_\text{mul}^m $}
		
		\State $\text{min}_\text{mul}^m  =   M^m_{i,d-1} + M^b_{d,e} + \Delta G_\text{bp}^\text{multi} + (j-e)\Delta G_\text{nt}^\text{multi}$
		
		\EndIf
		
		\EndIf
		
		\EndIf
		
		\EndFor
		\State $M_{i,j,e}^\text{m:2} = \text{min}_\text{mul}^m $
		\EndFor
		
		\EndFor
		\EndFor \Comment{next line returns the \snMFE for ordering $\pi$, and several matrices for future backtracking} 
		\State \Return $M_{1,N} + (c-1) \Delta G^\text{assoc}$; 
		and matrices: $M, M^b, M^m, M^\text{b:int}, M^\text{b:mul}$,  $M^\text{m:2}$

	\end{algorithmic}
\end{algorithm}


\section{Appendix: Backtracking algorithm to find the true MFE }\label{app:backalgo}

In this backtracking algorithm, we use  $[i,j]_{q:k}^t $ to denote the generic form of segment that has been popped from the stack $\delta$, where $t \in \{\square,b,m\}$, and $q \in \{\text{mul}, \text{int},  \text{null}\}$, where null means (nothing), and $k \in [i,j]$. For the full details, see \cref{sec:backhigh}.

\begin{algorithm}[t] 
	\caption{\small Backtracking pseudocode that takes as input: $c=\mathcal{O}(1)$ strands with total number of bases (length) $N$ and strand ordering $\pi$. 
		Runs in  $\mathcal{O}(N^4)$ time and $\mathcal{O}(N^4)$ space, and assumes there are $k \leq c$ strand types given as a \emph{multiset}, 
		each with an associated repetition number $n_1, ..., n_k \in \mathbb{N}$, such that $n_1+ ...+n_k = c$, with total length $N$. 
		$[i,j]^t \Leftarrow \delta$ denotes  popping an element from stack $\delta$, 
		which is a segment, and assigning it to the generic segment $[i,j]^t$.  
		And $\mathcal{S} \Rightarrow \mathcal{R}$ denotes pushing structure $\mathcal{S}$ onto stack $\mathcal{R}$, and  $E(\mathcal{S})$ is defined in \cref{eq:ES}, and all refinement cases are analyzed in \cref{sec:backhigh}.
	} \label{algo:2}
	\begin{algorithmic}[1]
		\footnotesize	
		\If{($n_1, n_2, \ldots, n_k$) are all even}\Comment{use \cref{lem:even} to set symmetric structure upperbound $\mathcal{U}$}
		\State $\mathcal{U} =  \frac{N-c}{v(\pi)} \left[ \sigma(v(\pi))-v(\pi) \right] + \frac{N^2}{16}$  \Comment{in $\mathcal{O}(1)$ time}
		
		\Else
		
		\State $\mathcal{U} =  \frac{N-c}{v(\pi)} \left[ \sigma(v(\pi))-v(\pi) \right]$
		
		\EndIf
		
		\State $\mathcal{E} =$ \snMFE \Comment{\snMFE is returned by \cref{algo:1}}
		
		\State $\mathcal{B} = \mathcal{E}+ k_\mathrm{B} T \log v(\pi)$ \Comment{where $v(\pi)$ is the highest symmetry degree for the $c$-strands ordering $\pi$}

		\State $\mathcal{S} = ([1,N]^\square, \phi,0)$ \Comment{the initial system (the parent of any possible structure)}
		
		\State $(\delta, \mathcal{P}, E_{L_{\mathcal{S}}}) = \mathcal{S}$
		
		\State $u = 1$ \Comment{$u$ is a symmetric secondary structures counter}
		
		\While{($u \leq \mathcal{U}$)}
		\State  $y= \mathrm{False}$; \ \ $w= \mathrm{False}$  \Comment{$y$ and $w$ are indicator variables}
		
		\State $[i,j]_{q:k}^t \Leftarrow \delta$
		
		\State $\mathcal{H} = E_{L_{\mathcal{S}}} + \sum \limits_{[m,n]^t \in \delta} E([m,n]^t)$\Comment{in $\mathcal{O}(N)$, \cref{remark:timeE}}\label{line:H}
		
		\State $\mathcal{S} = (\delta, \mathcal{P}, E_{L_{\mathcal{S}}})$ \Comment{if all cases were not satisfied (just pop  $[i,j]_{q:k}^t$)}

		\If{$(t == \square)$} \Comment{backtrack in matrix $M$}
		
		\State $\mathcal{S}' = ([i,j-1]^\square.\delta, \mathcal{P}, E_{L_{\mathcal{S}}})$\Comment{base $j$ is not paired}
		
		\State $E(\mathcal{S}')  = M_{i,j-1}  + \mathcal{H}$

		\If{($E(\mathcal{S}') \leq \mathcal{B}$)}	
		
		\If{($y== \mathrm{False}$ and $E(\mathcal{S}') == \mathcal{E}$)}
		
		\State $\mathcal{S} = \mathcal{S}'$; \ \  $y = \mathrm{True}$
		
		\Else
		
		\State $\mathcal{S}' \Rightarrow \mathcal{R}_u$		
		\EndIf

		\EndIf
		
		\For{$d \in [i,j-1]$}\Comment{base $j$ is paired with some base $d$}

		\State $\mathcal{S}' = ([i,d-1]^\square.[d,j]^b.\delta, \mathcal{P}, E_{L_{\mathcal{S}}})$
		
		\State $E(\mathcal{S}')  = M_{i,d-1} + M^b_{d,j}   + \mathcal{H} $
		
		\If{($ M_{i,d-1} + M^b_{d,j}   + \mathcal{H} \leq \mathcal{B}$)}
		
		\If{($y== F$ and $E(\mathcal{S}') == \mathcal{E}$)}
		
		\State $\mathcal{S} = \mathcal{S}'$; \ \  $y = T$
		
		\Else
		
		\State $\mathcal{S}' \Rightarrow \mathcal{R}_u$		
		
		\EndIf
		\EndIf
		\EndFor
		
		\ElsIf{$(t == b)$}\Comment{backtrack in matrix $M^b$}
		
		\If{($q == \text{null}$)}\Comment{hairpin loop formation}
		
		\State$\mathcal{S}' = (\delta, \mathcal{P} \cup \{(i,j)\}, E_{L_{\mathcal{S}}} + \Delta G_{i,j}^\text{hairpin} )$ 
		
		\State $E(\mathcal{S}')  = \Delta G_{i,j}^\text{hairpin} + \mathcal{H} $

		\If{($E(\mathcal{S}') \leq \mathcal{B}$)}	
		
		\If{($y== F$ and $E(\mathcal{S}') == \mathcal{E}$)}
		
		\State $\mathcal{S} = \mathcal{S}'$; \ \  $y = T$
		
		\Else
		
		\State $\mathcal{S}' \Rightarrow \mathcal{R}_u$		
		\EndIf

		\EndIf
		\EndIf

		\algstore{backtracking1}		
		
	\end{algorithmic}
\end{algorithm}

\begin{algorithm}
	\caption*{Part 2 of backtracking algorithm}
	\begin{algorithmic}[1]
		\algrestore{backtracking1}
		\footnotesize
		
		\If{($q == \text{null}$)} \Comment{interpreting the segment to be in a valid form for internal loop case \cref{sec:backhigh}.}
		
		\State $q = \text{int}$; \ \ $k = j$, \ \ $w= \mathrm{True}$
		\EndIf	
		
		\If{($q == \text{int}$)}
		
		\State$\mathcal{S}' = ([i,j]^b_{\text{int}:k-1}.\delta, \mathcal{P}, E_{L_{\mathcal{S}}})$ \Comment{internal loop formation: base $k-1$ is unpaired}
		
		\State $E(\mathcal{S}')  = M_{i,j,k-2}^\text{b:int} + \mathcal{H}$
		
		\If{($E(\mathcal{S}') \leq \mathcal{B}$)}	
		
		\If{($y== F$ and $E(\mathcal{S}') == \mathcal{E}$)}
		
		\State $\mathcal{S} = \mathcal{S}'$; \ \  $y = T$
		
		\Else \ \  $\mathcal{S}' \Rightarrow \mathcal{R}_u$		
		
		\EndIf
		\EndIf
		\EndIf

		\If{($w == \mathrm{True}$)}
		\State $q = \text{null}$; \ \ $k = \text{null}$, \ \ $w= \mathrm{False}$ 
		\EndIf

		\If{($q == \text{null}$)}\Comment{internal loop formation: base $k-1$ is paired}
		
		\For{$d \in [i+1,k-2]$}

		\State  $\mathcal{S}' = ([d,k-1]^b.\delta, \mathcal{P} \cup \{(i,j)\}, E_{L_{\mathcal{S}}} + \Delta G_{i,d,k-1,j}^\text{interior})$
		
		\State $E(\mathcal{S}')  = M^b_{d,k-1} + \Delta G_{i,d,k-1,j}^\text{interior} + \mathcal{H}$

		\If{($E(\mathcal{S}') \leq \mathcal{B}$)}	
		
		\If{($y== F$ and $E(\mathcal{S}') == \mathcal{E}$)}
		
		\State $\mathcal{S} = \mathcal{S}'$; \ \  $y = T$
		
		\Else  \ \  $\mathcal{S}' \Rightarrow \mathcal{R}_u$

		\EndIf

		\EndIf
		\EndFor
		\EndIf
		
		\If{($q == \text{null}$)}\Comment{interpreting the segment to be in a valid form for multiloop case \cref{sec:backhigh}}
		
		\State $q = \text{mul}$; \ \ $k = j$, \ \ $w= \mathrm{True}$
		\EndIf	
		
		\If{($q == \text{mul}$)}
		
		\State $\mathcal{S}' = ([i,j]^b_{\text{mul}:k-1}.\delta, \mathcal{P}, E_{L_{\mathcal{S}}})$ \Comment{multiloop formation: base $k-1$ is unpaired}
		
		\State $E(\mathcal{S}')  = M_{i,j,k-2}^\text{b:mul} + \mathcal{H}$

		\If{($E(\mathcal{S}') \leq \mathcal{B}$)}	
		
		\If{($y== F$ and $E(\mathcal{S}') == \mathcal{E}$)}
		
		\State $\mathcal{S} = \mathcal{S}'$; \ \  $y = T$
		
		\Else
		
		\State $\mathcal{S}' \Rightarrow \mathcal{R}_u$		
		\EndIf
		
		\EndIf
		
		\EndIf
		\If{($w == \mathrm{True}$)}
		\State $q = \text{null}$; \ \ $k = \text{null}$, \ \ $w= \mathrm{False}$ 
		\EndIf	
		
		\If{($q == \text{null}$)}\Comment{multiloop formation: base $k-1$ is paired}
		
		\For{$d \in [i+1,k-2]$}

		\State  $\mathcal{S}' = ([i+1,d-1]^m.[d,k-1]^b.\delta, \mathcal{P} \cup \{(i,j)\}, \Delta G_\text{init}^\text{multi} + 2\Delta G_\text{bp}^\text{multi} + (j-k) \Delta G_\text{nt}^\text{multi} + E_{L_{\mathcal{S}}})$
		
		\State $E(\mathcal{S}')  = M^m_{i+1,d-1} + M^b_{d,k-1} + \Delta G_\text{init}^\text{multi} + 2\Delta G_\text{bp}^\text{multi} + (j-k)\Delta G_\text{nt}^\text{multi} + \mathcal{H}$

		\If{($ E(\mathcal{S}') \leq \mathcal{B}$)}	
		
		\If{($y== F$ and $E(\mathcal{S}') == \mathcal{E}$)}
		
		\State $\mathcal{S} = \mathcal{S}'$; \ \  $y = T$
		
		\Else
		
		\State $\mathcal{S}' \Rightarrow \mathcal{R}_u$		
		\EndIf

		\EndIf
		\EndFor
		\EndIf

		\algstore{backtracking2}
		
	\end{algorithmic}
\end{algorithm}

\begin{algorithm}
	\caption*{Part 3 of backtracking algorithm}
	\begin{algorithmic}[1]
		\algrestore{backtracking2}
		\footnotesize
		
		\If{($q == \text{null}$)}\Comment{exterior loop formation}
		
		\For{$z \in [i,j]$ s.t. $\eta[z+\frac{1}{2}] ==1$}

		\State  $\mathcal{S}' = ([i+1,z]^\square.[z+1,j-1]^\square.\delta, \mathcal{P} \cup \{(i,j)\}, E_{L_{\mathcal{S}}})$
		
		\State $E(\mathcal{S}')  = M_{i+1,z} +  M_{z+1,j-1} + \mathcal{H}$

		\If{($E(\mathcal{S}') \leq \mathcal{B}$)}	
		
		\If{($y== F$ and $E(\mathcal{S}') == \mathcal{E}$)}
		
		\State $\mathcal{S} = \mathcal{S}'$; \ \  $y = T$
		
		\Else
		
		\State $\mathcal{S}' \Rightarrow \mathcal{R}_u$		
		\EndIf
		\EndIf
		\EndFor
		
		\EndIf

		\ElsIf{$(t == m)$}\Comment{backtrack in matrix $M^m$}
		
		\If{($q == \text{null}$)}\Comment{multiloop case 1: base $j$ is unpaired}
		
		\State $\mathcal{S}' = ([i,j-1]^m.\delta, \mathcal{P}, E_{L_{\mathcal{S}}} + \Delta G_\text{nt}^\text{multi})$ 
		
		\State $E(\mathcal{S}')  = M_{i,j-1}^m + \Delta G_\text{nt}^\text{multi} + \mathcal{H}$

		\If{($E(\mathcal{S}')  \leq \mathcal{B}$)}	
		
		\If{($y== F$ and $E(\mathcal{S}') == \mathcal{E}$)}
		
		\State $\mathcal{S} = \mathcal{S}'$; \ \  $y = T$
		
		\Else
		
		\State $\mathcal{S}' \Rightarrow \mathcal{R}_u$		
		\EndIf
		
		\EndIf
		
		\EndIf
		\If{($q == \text{null}$)}\Comment{base $j$ is paired}
		
		\For{$d \in [i,j-1]$}

		\State $\mathcal{S}' = ([d,j]^b.\delta, \mathcal{P}, E_{L_{\mathcal{S}}} + \Delta G_\text{bp}^\text{multi} + (d-i) \Delta G_\text{nt}^\text{multi})$ 
		
		\State $E(\mathcal{S}') = M^b_{d,j} + \Delta G_\text{bp}^\text{multi}+ (d-i) \Delta G_\text{nt}^\text{multi} + \mathcal{H}$

		\If{($E(\mathcal{S}')  \leq \mathcal{B}$)}	
		
		\If{($y== F$ and $E(\mathcal{S}') == \mathcal{E}$)}
		
		\State $\mathcal{S} = \mathcal{S}'$; \ \  $y = T$
		
		\Else
		
		\State $\mathcal{S}' \Rightarrow \mathcal{R}_u$		
		\EndIf

		\EndIf

		\EndFor
		\EndIf
		
		\If{($q == \text{null}$)}\Comment{interpreting the segment to be in a valid form for multiloop case \cref{sec:backhigh}}
		
		\State $q = \text{mul}$; \ \ $k = j-1$, \ \ $w= \mathrm{True}$
		\EndIf	
		
		\If{($q == \text{mul}$)}\Comment{multiloop case 2: base $k$ is unpaired}
		
		\State  $\mathcal{S}' = ([i,j]^m_{\text{mul}:k-1}.\delta, \mathcal{P}, E_{L_{\mathcal{S}}})$ 
		
		\State $E(\mathcal{S}') = M_{i,j,k-2}^\text{m:2} + \mathcal{H}$
		
		\If{($E(\mathcal{S}') \leq \mathcal{B}$)}	
		
		\If{($y== F$ and $E(\mathcal{S}') == \mathcal{E}$)}
		
		\State $\mathcal{S} = \mathcal{S}'$; \ \  $y = T$
		
		\Else
		
		\State $\mathcal{S}' \Rightarrow \mathcal{R}_u$		
		\EndIf
		
		\EndIf
		
		\EndIf
		
		\If{($w == \mathrm{True}$)}
		\State $q = \text{null}$; \ \ $k = \text{null}$, \ \ $w= \mathrm{False}$ 
		\EndIf

		\algstore{backtracking3}
		
	\end{algorithmic}
\end{algorithm}

\begin{algorithm}
	\caption*{Part 4 of backtracking algorithm}
	\begin{algorithmic}[1]
		\algrestore{backtracking3}
		\footnotesize
		
		\If{($q == \text{null}$)}\Comment{base $k-1$ is paired}
		
		\For{$d \in [i,k-2]$}
		
		\State $\mathcal{S}' = ([i,d-1]^m.[d,k-1]^b.\delta, \mathcal{P},  \Delta G_\text{bp}^\text{multi} + (j-k+1) \Delta G_\text{nt}^\text{multi} + E_{L_{\mathcal{S}}})$

		\State $E(\mathcal{S}') = M^m_{i,d-1} +  M^b_{d,k-1} + \Delta G_\text{bp}^\text{multi} + (j-k + 1)\Delta G_\text{nt}^\text{multi} + \mathcal{H} $
		
		\If{($ E(\mathcal{S}') \leq \mathcal{B}$)}	
		
		\If{($y== F$ and $E(\mathcal{S}') == \mathcal{E}$)}
		
		\State $\mathcal{S} = \mathcal{S}'$; \ \  $y = T$
		
		\Else
		
		\State $\mathcal{S}' \Rightarrow \mathcal{R}_u$		
		\EndIf

		\EndIf
		\EndFor
		\EndIf
		\EndIf
		
		\State $(\delta, \mathcal{P}, E_{L_{\mathcal{S}}}) = \mathcal{S}$
		
		\If{$(\delta == \phi)$}\Comment{ scanning $\mathcal{S}$ is done, $\mathcal{S}$ is a fully specified structure}
		
		\State Output secondary structure $\mathcal{S}$ using its base pairs set $\mathcal{P}$
		
		\If{($\mathcal{S}$ is asymmetric)}
		\State  MFE $= \Delta G(\mathcal{S})$  \Comment{$\Delta G(\mathcal{S})$ is defined in \cref{eq:DGss}, which also equals $\mathcal{E}$ at the moment} 
		\Break  \Comment{break out of top-level while loop}
		\Else
		\State Apply rot.~symmetry correction to free energy of $\mathcal{S}$ (i.e.~$\mathcal{E}' := \mathcal{E} + k_\mathrm{B} T \log R$);  
		if $\mathcal{E}'<\mathcal{B}$ then $\mathcal{B} := \mathcal{E}'$

		\State  $\mathcal{S} = \min\limits_{z \in \left\{1,\ldots,u\right\}}\left\{\min\limits_{\mathcal{S}' \in \mathcal{R}_z}\left\{ E(\mathcal{S}')\right\}\right\}$  \Comment{\cref{eq:newS}, to ensure the sequential scanning over energy levels} 
		
		\State	$u = u+1$   \Comment{increment symmetric secondary structures counter}

		\State $(\delta, \mathcal{P}, E_{L_{\mathcal{S}}}) = \mathcal{S}$
		
		\If{($E(\mathcal{S}) > \mathcal{B}$)}\Comment{compare with the updated $\mathcal{B}$}

		\State    MFE $=  \mathcal{B}$ 
		\Break  \Comment{break out of top-level while loop}
		
		\Else
		\State $\mathcal{E} = E(\mathcal{S})$
		
		\State    MFE $=  \mathcal{E}$ \Comment{in case the upper bound $\mathcal{U}$ is exceeded}
		
		\EndIf		
		
		\EndIf	
		\EndIf
		
		\EndWhile

		\State \Return MFE   \Comment{return the true MFE}

	\end{algorithmic}
\end{algorithm}

\end{document}